\documentclass[onefignum,onetabnum]{siamonline171218}
\pdfoutput=1

\usepackage{lipsum}
\usepackage{amsfonts,amsmath}
\usepackage{graphicx}
\usepackage{epstopdf}
\usepackage{algorithmic}
\ifpdf
  \DeclareGraphicsExtensions{.eps,.pdf,.png,.jpg}
\else
  \DeclareGraphicsExtensions{.eps}
\fi

\usepackage{enumitem}
\setlist[enumerate]{leftmargin=.5in}
\setlist[itemize]{leftmargin=.5in}

\newsiamremark{remark}{Remark}
\newsiamremark{hypothesis}{Hypothesis}
\crefname{hypothesis}{Hypothesis}{Hypotheses}
\newsiamthm{claim}{Claim}

\headers{Multifidelity ABC}{T. P. Prescott and R. E. Baker}

\title{Multifidelity Approximate Bayesian Computation. \thanks{Submitted to the editors: 29 Nov 2018.
}
}

\author{Thomas P. Prescott\thanks{Wolfson Centre for Mathematical Biology, Mathematical Institute, University of Oxford, Andrew Wiles Building, Woodstock Road, Oxford, OX2 6GG, U.K.
  (\email{prescott@maths.ox.ac.uk}, \email{baker@maths.ox.ac.uk}).}
\and Ruth E. Baker\footnotemark[2]}

\usepackage{amsopn}

\usepackage{setspace}
\DeclareMathOperator*{\argmin}{arg\,min}

\ifpdf
\hypersetup{
  pdftitle={Multifidelity ABC},
  pdfauthor={T. P. Prescott and R. E. Baker}
}
\fi

\begin{document}

\maketitle

\begin{abstract}
A vital stage in the mathematical modelling of real-world systems is to calibrate a model's parameters to observed data. 
Likelihood-free parameter inference methods, such as approximate Bayesian computation (ABC), build Monte Carlo samples of the uncertain parameter distribution by comparing the data with large numbers of model simulations. 
However, the computational expense of generating these simulations forms a significant bottleneck in the practical application of such methods. 
We identify how simulations of corresponding cheap, low-fidelity models have been used separately in two complementary ways to reduce the computational expense of building these samples, at the cost of introducing additional variance to the resulting parameter estimates. 
We explore how these approaches can be unified so that cost and benefit are optimally balanced, and we characterise the optimal choice of how often to simulate from cheap, low-fidelity models in place of expensive, high-fidelity models in Monte Carlo ABC algorithms. 
The resulting early accept/reject multifidelity ABC algorithm that we propose is shown to give improved performance over existing multifidelity and high-fidelity approaches.
\end{abstract}



\section{Introduction}
\label{s:intro}


Throughout all scientific domains, predictive models of complex dynamical systems require calibration against experimental data.
%
Approximate Bayesian computation (ABC) is a popular likelihood-free method of parameter inference for complex models in the biomedical sciences~\cite{Sisson2018}. 
Rather than calculating the likelihood of the data for any given parameter, the predictive model is simulated using that parameter. 
The likelihood is then estimated based on how close, in some sense, the observed data is to the simulated data.
A classical technique is known as rejection sampling, where the likelihood is approximated with a randomly assigned value of $1$ (accept) or $0$ (reject), where the probability of acceptance is larger for simulations that are close to the data.
The prior parameter distribution is explored by repeatedly evaluating this accept/reject decision for a large number of parameter values sampled from the prior.
Therefore, ABC sampling typically requires a large number of simulations, which can form a bottleneck if the computational cost of each simulation is prohibitively high.

The efficiency of ABC can be improved using parallelisation~\cite{Jagiella2017}, or with alternative sampling strategies that reduce the number of required simulations by a more efficient exploration of the prior distribution.
These include Markov chain Monte Carlo (MCMC)~\cite{Marjoram2003} and Sequential Monte Carlo (SMC)~\cite{DelMoral2012,Toni2009} approaches, which ensure that simulated parameters are sampled more often from high-likelihood regions of parameter space.
A wider discussion of these sampling strategies can be found in~\cite[ch.~4]{Sisson2018}.
Although the parameter space is explored more efficiently with these methods, there remains a high computational burden from a large number of repeated simulations.
Rather than focusing on exploring parameter space efficiently, this paper instead focuses on reducing the computational burden of the Monte Carlo sampling approach by using models that can be simulated more cheaply.

In this work, we consider a model as a map from a parameter vector to a distribution on an output space.
To simulate a model is to draw from the output distribution, the computational burden of which is the simulation cost.
Note that our use of `model' includes domain-specific modelling choices and numerical implementation.
Many ways to approximate a given model with one that can be simulated more cheaply have been proposed and investigated, such as model reduction~\cite{Antoulas2005,Benner2015,Benner2017,Snowden2017}, discretisation~\cite{Gillespie2001}, surrogate modelling~\cite{Rasmussen2006}, and early stopping~\cite{Jiang2012}.
Recent work~\cite{Peherstorfer2018,Peherstorfer2016} unifies these approaches in the context of \emph{multifidelity} methods, which integrate information from many models of the same system to accelerate tasks such as optimisation, inference, and uncertainty quantification.
Here, we use the terminology of Peherstorfer et al.~\cite{Peherstorfer2016}, denoting the model being calibrated as the high-fidelity model, and other models as low-fidelity models.
Simulations from low/high-fidelity models are termed low/high-fidelity simulations: we assume that low-fidelity simulations are cheaper than high-fidelity simulations.

Multilevel Monte Carlo (MLMC)~\cite{Giles2008,Giles2015} is one example of a multifidelity estimation approach. 
In its original formulation, continuous-time stochastic differential equations are simulated using progressively finer, more accurate, discretisations.
For a given computational budget, the statistical error of a Monte Carlo estimate can be reduced by using variance reduction techniques that combine estimates built from simulations at different discretisations with common input noise.
The key aim of MLMC implementation is to optimise the number of simulations using each of the different discretisations to reduce the estimator's variance.

Previous work has exploited multifidelity approaches to parameter inference~\cite{Cui2015,Cui2016}.
A multilevel approach to ABC is considered in \cite{Warne2018}, where a set of ABC samples of increasing simulation cost is produced by using progressively stricter rejection sampling thresholds, chosen to optimise the efficiency of building the overall sample.
In approximate ABC~\cite{Buzbas2015} (aABC) a small number of simulations are used to create a low-fidelity statistical surrogate of the model output across parameter space, to which ABC is applied.
Other examples include Lazy ABC~\cite{Prangle2016} and Delayed Acceptance ABC~\cite{Christen2005}, where low-fidelity simulations are used to decide whether the parameter can be rejected, without necessarily needing to simulate from a high-fidelity model. 


In this paper we apply multifidelity model management ideas to the specific case of rejection sampling ABC.
We present a new method that allows a reject/accept decision to be made for a parameter sample using a low-fidelity simulation alone, i.e. without necessarily requiring a corresponding high-fidelity simulation.
\Cref{s:abc} introduces ABC and the motivation for multifidelity approaches.
We develop these into a new multifidelity rejection sampling algorithm in \Cref{s:multifidelity}. 
In \Cref{ss:comparing} we describe how to analyse the performance of this algorithm and optimise its inputs. 
The theoretical work is illustrated by applying the multifidelity rejection sampling algorithm to a stochastic synthetic biology model in \Cref{s:repressilator}, which is also used in \Cref{s:implementation} to illustrate practical issues around implementation. 
We consider a second example in \Cref{ss:viral}, and conclude with a view of potential future developments in \Cref{s:end}.
Our code, implemented in Julia, is available at \url{https://github.com/tpprescott/mf-abc}.

\section{ABC and estimators}
\label{s:abc}

The goal of Bayesian parameter estimation is to update prior beliefs about model parameters, $\theta$, encoded in a prior distribution $\pi(\theta)$.
The updates depend on experimental observations, $y_{\mathrm{obs}}$, subject to stochasticity such as measurement and environmental noise.
The parameterised model is denoted $p(\cdot~|~\theta)$, which defines a likelihood, $p(y_{\mathrm{obs}}~|~\theta)$.
The likelihood is combined with the prior distribution to give the posterior distribution, $p(\theta ~|~ y_{\mathrm{obs}}) \propto p(y_{\mathrm{obs}} ~|~ \theta) \pi(\theta)$.
We assume that the likelihood is not available, and that we need to use ABC to estimate the posterior distribution. 

The simplest version of ABC approximates the likelihood, $p(y_{\mathrm{obs}}~|~\theta)$, of observing $y_{\mathrm{obs}}$ under the model, based on the simulations of the model being in some sense \emph{close enough} to $y_{\mathrm{obs}}$.
This gives the approximate posterior,
\begin{equation}
p_{\mathrm{ABC}}(\theta ~|~ y_{\mathrm{obs}}) = \frac{p(d(y,y_{\mathrm{obs}})<\epsilon ~|~ \theta) \pi(\theta)}{Z}
= \frac{ p(y \in \Omega(\epsilon) ~|~ \theta) \pi(\theta)}{Z} ,
\label{eq:abcposterior}
\end{equation}
where the normalisation constant $Z$ ensures the distribution has unit integral, and $\Omega(\epsilon) = \{ y~|~ d(y,y_{\mathrm{obs}})<\epsilon) \}$ is the $\epsilon$-close neighbourhood of $y_{\mathrm{obs}}$, where $d(y,y_{\mathrm{obs}})$ is a distance measure between the observations, $y_{\mathrm{obs}}$, and model outputs, $y$. 
The approximate posterior also induces an expectation,
\[
\mathbb E_{\mathrm{ABC}}(F(\theta)~|~y_{\mathrm{obs}}) = \int F(\theta)~ p_{\mathrm{ABC}}(\theta~|~y_{\mathrm{obs}}) ~\mathrm d\theta,
\]
which is the ABC approximation to the posterior expectation of an arbitrary function $F$.

The value of $p(y \in \Omega(\epsilon) ~|~ \theta)$ is typically estimated using simulation.
Given $\theta \sim \pi(\cdot)$ sampled from the prior distribution, we simulate $y \sim p(\cdot ~|~ \theta)$ from the model and calculate a weight $w(\theta) = \mathbb I( y \in \Omega(\epsilon))$. 
If we consider $w(\theta)=0$ as rejection and $w(\theta)=1$ as acceptance of $\theta$, the parameter is accepted (resp. rejected) if it generates summary statistics that are close to (resp. far from) the observed data.
Taking the expectation over $y \sim p(\cdot ~|~ \theta)$ gives $\mathbb E(w(\theta)~|~\theta) = p( y \in \Omega(\epsilon) ~|~ \theta)$.
Thus, $w(\theta)$ is an unbiased estimator of the ABC approximation to the likelihood.

\begin{algorithm}
 \caption{Rejection sampling ABC}
 \label{a:abc}
 \begin{algorithmic}
  \STATE{Input: observed measurements $y_{\mathrm{obs}}$; prior $\pi(\cdot)$; function $F(\theta)$; model $p(\cdot~|~\theta)$; distance function $d(\cdot,y_{\mathrm{obs}})$; threshold $\epsilon$; Monte Carlo sample size $N$.}
  \STATE{~}
  \FOR{$i = 1,\dots,N$}
   \STATE{Generate $\theta_i \sim \pi(\cdot)$.}
   \STATE{Simulate $y_i \sim p(\cdot~|~\theta_i)$ from the model.}
   \STATE{Calculate $w_i = w(\theta_i) = \mathbb I(d(y_i,y_{\mathrm{obs}}) < \epsilon)$.}
  \ENDFOR{}
  \STATE{Calculate $\mu_{\mathrm{ABC}}(F) = \sum_{i=1}^N w_i F(\theta_i) / \sum_{j=1}^N w_j$.}
  \RETURN{$\{w_i,\theta_i\}$ and $\mu_{\mathrm{ABC}}(F)$.}
 \end{algorithmic}
\end{algorithm}

The weights $w(\theta)$ can be used in a Monte Carlo algorithm to build a weighted sample $\{w_i, \theta_i\}$.
The simplest approach, the ABC Rejection Sampler (\cref{a:abc}) involves independently generating $\theta_i \sim \pi(\cdot)$ for $i=1,\dots,N$, and setting $w_i = w(\theta_i)$. 
The estimator calculated by \cref{a:abc} is
\begin{equation}
\label{eq:estimator}
 \mu_{\mathrm{ABC}}(F) = \frac{\sum_i w(\theta_i) F(\theta_i) / N}{\sum_j w(\theta_j)/N} \approx \frac{Z \mathbb E_{\mathrm{ABC}}(F(\theta)~|~y_{\mathrm{obs}})}{Z} =  \mathbb E_{\mathrm{ABC}}(F(\theta)~|~y_{\mathrm{obs}}).
\end{equation}
The numerator and denominator of $\mu_{\mathrm{ABC}}$ are each unbiased estimators of $Z \mathbb E_{\mathrm{ABC}}(F(\theta)~|~y_{\mathrm{obs}})$ and $Z$, respectively. 
Although the ratio is not an unbiased estimator of $\mathbb E_{\mathrm{ABC}}(F~|~y_{\mathrm{obs}})$, the bias of $\mu_{\mathrm{ABC}}(F)$ vanishes as the sample size $N$ becomes large~\cite{Sisson2018}.

The key issue with \cref{a:abc} is that a large number, $N$, of simulations $y_i \sim p(\cdot~|~\theta_i)$ are required to generate an accurate approximation of $\mathbb E_{\mathrm{ABC}}(F(\theta)~|~y_{\mathrm{obs}})$. 
%
Rather than aiming to reduce the number, $N$, of simulations~\cite{Marjoram2003,Toni2009}, this paper considers the use of computationally cheap approximations to $w(\theta) = \mathbb I(y \in \Omega(\epsilon))$. 
The goal is to reduce the computational burden of producing $\mu_{\mathrm{ABC}}(F)$ for any fixed number, $N$, of Monte Carlo sample points. 

\section{Multifidelity approximate Bayesian computation}
\label{s:multifidelity}

To reduce the computational cost of rejection sampling ABC, we will exploit the concept of multifidelity modelling~\cite{Peherstorfer2018}. 
The high-fidelity `ground truth' model, $p(\cdot~|~\theta)$, is assumed to be a computationally expensive, accurate representation of the observed system.
We consider the model to be a map from parameter sample $\theta$ to a distribution on an output space containing the observations, $y_{\mathrm{obs}}$.
A simulation from the high-fidelity model (i.e. a high-fidelity simulation) for a particular $\theta$ is a draw $y \sim p(\cdot~|~\theta)$ from this distribution, the computational cost of which is denoted by $c(\theta)$.

We also consider a \emph{low-fidelity} model, $\tilde p(\cdot~|~\theta)$, which is an alternative map from the parameter sample $\theta$ to a distribution on an output space.
Note that the output space of the low-fidelity model may be different to that of the high-fidelity model:
we assume that the output space is induced by taking potentially different measurements $\tilde y_{\mathrm{obs}}$ from the same experiment generating the measurements comprising $y_{\mathrm{obs}}$.
A simulation from the low-fidelity model (i.e. a low-fidelity simulation) is a draw $\tilde y \sim \tilde p(\cdot~|~\theta)$, the computational cost of which is denoted $\tilde c(\theta)$.
We will assume that low-fidelity simulations are, on average, much cheaper than high-fidelity simulations, such that $\mathbb E(\tilde c(\theta)) \ll \mathbb E(c(\theta))$.
In direct analogy with \Cref{a:abc}, we define a distance function $\tilde d(\tilde y, \tilde y_{\mathrm{obs}})$, measuring how close the simulated data is to the observed data, and a threshold $\tilde \epsilon$.
These define a weight $\tilde w(\theta) = \mathbb I(\tilde y \in \tilde \Omega(\tilde \epsilon))$, where we write $\tilde \Omega(\tilde \epsilon) = \{ \tilde y ~|~ \tilde d(\tilde y,\tilde y_{\mathrm{obs}})<\tilde \epsilon) \}$ for the neighbourhood of the data. 

The sample $\{ \tilde w(\theta_i), \theta_i \}$ will be built more quickly than $\{ w(\theta_i),\theta_i \}$, for a fixed $N$.
However, this computational speedup comes at the cost of bias, which arises because the likelihood of the low-fidelity model does not equal that of the high-fidelity model. 
The ABC approximations to each likelihood are also not identical, since
$
 \mathbb E(\tilde w(\theta)) = \tilde p(\tilde y \in \tilde \Omega(\tilde \epsilon)~|~\theta) \neq p(y \in \Omega(\epsilon)~|~\theta) = \mathbb E(w(\theta))
$.
The bias is compounded by the fact that the observations $y_{\mathrm{obs}}$ and $\tilde y_{\mathrm{obs}}$, distance functions $d$ and $\tilde d$, and thresholds $\epsilon$ and $\tilde \epsilon$, may be specified independently of one another.

The goal of the remainder of this section is to consider how best to use the information generated by the low-fidelity model to reduce the reliance on the high-fidelity model in estimating $p_{\mathrm{ABC}}(\theta~|~y_{\mathrm{obs}})$.
We aim to produce an unbiased estimate of the ABC approximation to the likelihood generated by the high-fidelity model, $p(y \in \Omega(\epsilon)~|~\theta) \approx p(y_{\mathrm{obs}}~|~\theta)$.

\subsection{Early rejection ABC}
\label{ss:lz}

As a starting point, we will describe an existing approach that uses the low-fidelity model, $\tilde p(\cdot~|~\theta)$, to reduce the cost of calculating an unbiased estimator of $p(y \in \Omega(\epsilon)~|~\theta)$.
A version of this approach is used in \emph{lazy ABC}~\cite{Prangle2016} and is also the key idea of \emph{delayed acceptance MCMC}~\cite{Christen2005}, but here we will refer to it as \emph{early rejection} ABC.
Recall that the weight $w(\theta)$ is an unbiased estimator of the ABC approximation to the likelihood, $p(y \in \Omega(\epsilon)~|~\theta)$, and requires a simulation of the high-fidelity model.
The early rejection ABC approach generates an alternative unbiased estimator, $w_{\mathrm{er}}(\theta)$, which saves computational costs by using the result of the low-fidelity simulation to decide whether to simulate the high-fidelity model, or reject the parameter early.

For a sample $\theta$ from the prior, we first simulate $\tilde y  \sim \tilde p(\cdot~|~\theta)$ from the low-fidelity model at a cost $\tilde c$.
A continuation probability $\eta(\tilde y) \in (0,1]$ is then defined, dependent on the result of the low-fidelity simulation.
With probability $1-\eta(\tilde y)$, the parameter is \emph{rejected early}: 
without simulating $y \sim p(\cdot~|~\theta)$, and therefore avoiding simulation cost $c$, the weight is set to $w_{\mathrm{er}} = 0$.
Otherwise, the high-fidelity simulation $y \sim p(\cdot~|~\theta)$ is generated and the parameter is accepted or rejected according to $\mathbb I(y \in \Omega(\epsilon))$, as before. 
If accepted, however, the weight is set to $w_{\mathrm{er}} = 1/\eta(\tilde y)$ rather than $1$.
For the uniform random variable $U \sim \mathrm{Unif}(0,1)$, we can write
\begin{equation}
\label{eq:lzweight}
 w_{\mathrm{er}}(\theta) = \frac{\mathbb I(U < \eta(\tilde y))}{\eta(\tilde y)} \mathbb I(y \in \Omega(\epsilon)).
\end{equation}
Taking the expectation with respect to $U$ recovers $w(\theta)$, and hence $\mathbb E(w_{\mathrm{er}}(\theta)) = \mathbb E(w(\theta)) = p(y \in \Omega(\epsilon)~|~\theta)$. 
Thus the early rejection estimate, $w_{\mathrm{er}}(\theta)$, is unbiased.

The improved performance of early rejection ABC relies on the low-fidelity simulation output, $\tilde y$, being informative about the high-fidelity simulation output, $y$, and on the careful definition of the continuation probabilities $\eta(\tilde y)$.
Firstly, as we will show in \Cref{ss:essefficiency}, the expected time taken to compute $w_{\mathrm{er}}(\theta)$ is less than for $w(\theta)$ if $\mathbb E(\eta(\tilde y)) < 1-\mathbb E(\tilde c)/\mathbb E(c)$.
Furthermore, suppose that $\tilde y$ is such that, with high probability, $y \notin \Omega(\epsilon)$ and hence $\theta$ will be rejected.
Rather than generate $\mathbb I(y \in \Omega(\epsilon))$ at cost $c$, it would be preferable to reject $\theta$ early.
For such $\tilde y$, this is achieved by ensuring $\eta(\tilde y)$ is small.
Conversely, if $\tilde y$ is such that, with high probability, $y \in \Omega(\epsilon)$ then $\theta$ is more likely to be accepted, corresponding to a positive value of $w_{\mathrm{er}}(\theta)$. 
It follows that $\eta(\tilde y)$ should be larger, allowing a positive weight, meaning that $y$ is more likely to be generated. 
However, the converse uncovers an important asymmetry underlying the early rejection approach. 
If $\tilde y$ is such that $y \in \Omega(\epsilon)$ with high probability, then an efficient approach could be to assign a positive weight to $\theta$ without simulating $y \sim p(\cdot~|~\theta)$ from the high-fidelity model.
However, such an \emph{early acceptance} is not possible within the framework of early rejection.

\subsection{Early decision ABC}
\label{ss:multilevel}

Instead of using $\tilde y$ to determine whether or not to simulate the high-fidelity model, we now assume that this decision is independent of $\tilde y$.
We can instead use $\tilde y$ to determine the weight for $\theta$ if the high-fidelity model is not simulated.
As with early rejection, for a given $\theta$ we first simulate $\tilde y \sim \tilde p(\cdot~|~\theta)$ from the low-fidelity model. 
Now suppose a continuation probability $\eta \in (0,1]$ is fixed (independently of $\tilde y$).
Then, with probability $1-\eta$, the parameter $\theta$ is accepted or rejected based on the early decision, $\mathbb I(\tilde y \in \tilde \Omega(\tilde \epsilon))$, without simulating the high-fidelity model and thus avoiding cost $c$.
Otherwise, with probability $\eta$, we simulate $y \sim p(\cdot~|~\theta)$ from the high-fidelity model and calculate $\mathbb I(y \in \Omega(\epsilon))$ to determine acceptance or rejection, as before. 
The appropriate weight for $\theta$ is
\begin{equation}
\label{eq:mlweight}
 w_{\mathrm{ed}}(\theta) = \mathbb I(\tilde y \in \tilde \Omega(\tilde \epsilon)) + \frac{\mathbb I(U<\eta)}{\eta} \left[ \mathbb I( y \in  \Omega( \epsilon)) - \mathbb I(\tilde y \in \tilde \Omega(\tilde \epsilon)) \right],
\end{equation}
where, again, taking the expectation over $U \sim \mathrm{Unif}(0,1)$ recovers $w(\theta)$. 
Thus, $w_{\mathrm{ed}}(\theta)$ is another unbiased estimator for $p(y \in \Omega(\epsilon)~|~\theta)$. 
Note that we can consider $w_{\mathrm{ed}}(\theta)$ as a multilevel weight, since it is a randomised multilevel estimator for $p(y \in \Omega(\epsilon)~|~\theta)$~\cite{Rhee2015}. 

Note that, by allowing early acceptance, $w_{\mathrm{ed}}(\theta)$ can take negative values.
In particular, if we simulate $U \leq \eta$, $\tilde y \in \tilde \Omega(\tilde \epsilon)$, and $y \notin \Omega(\epsilon)$, then $w_{\mathrm{ed}} = 1 - 1/\eta \leq 0$. 
This is a necessary consequence of early acceptance, which may overestimate the posterior weight on $\theta$ where $\tilde y \in \tilde \Omega(\tilde \epsilon)$. 
Negative weights means that the constructed set $\{ w_{\mathrm{ed}}(\theta_i), \theta_i \}$ cannot be interpreted as a weighted sample from the ABC posterior. 
Nevertheless, it is still valid to use $\{ w_{\mathrm{ed}}(\theta_i), \theta_i \}$ in the estimator $\mu_{\mathrm{ABC}}(F)$.

\subsection{Multifidelity ABC: early acceptance and early rejection}
\label{ss:rates}

We are now in a position to introduce early accept/reject multifidelity ABC. 
The approaches discussed in \Cref{ss:lz,ss:multilevel} use the low-fidelity simulation output, $\tilde y \sim \tilde p(\cdot~|~\theta)$, in different ways. 
The early rejection weight $w_{\mathrm{er}}(\theta)$ uses $\tilde y$ to determine whether to simulate the high-fidelity model.
In contrast, when calculating the early decision weight $w_{\mathrm{ed}}(\theta)$, we determine whether to simulate the high-fidelity model independently of $\tilde y$.
However, $w_{\mathrm{ed}}$ uses $\tilde y$ to determine the early decision that is to be made (i.e. accept or reject $\theta$) if the high-fidelity model is not simulated.
The following expression combines these ideas in a, more general, \emph{multifidelity} weight,
\begin{equation}
\label{eq:slweight}
 w_{\mathrm{mf}}(\theta) = \mathbb I(\tilde y \in \tilde \Omega(\tilde \epsilon)) + \frac{\mathbb I(U<\eta(\tilde y))}{\eta(\tilde y)} \left[ \mathbb I( y \in  \Omega( \epsilon)) - \mathbb I(\tilde y \in \tilde \Omega(\tilde \epsilon)) \right],
\end{equation}
where the continuation probability and early decision both depend on the output of the low-fidelity simulation.

As with early rejection ABC, the choice of continuation probability $\eta(\tilde y)$ is important to the performance of $w_{\mathrm{mf}}$.
A natural form of continuation probability, and the one we consider here, is
\begin{equation}
 \eta(\tilde y) =  \eta_1 \mathbb I( \tilde y \in \tilde \Omega(\tilde \epsilon)) + \eta_2 \mathbb I( \tilde y \notin \tilde \Omega(\tilde \epsilon)).
 \label{eq:slweight_contprob}
\end{equation}
This choice of $\eta(\tilde y)$ allows both early acceptance and early rejection with constant probabilities $1-\eta_1$ and $1-\eta_2$, respectively.
We will therefore refer to using $w_{\mathrm{mf}}$ and $\eta(\tilde y)$ given by \Cref{eq:slweight,eq:slweight_contprob} as \emph{early accept/reject multifidelity ABC}.
Note that constraining $\eta_1 = \eta_2$ makes $\eta$ independent of $\tilde y$ and recovers the early decision weight $w_{\mathrm{ed}}$. 
Fixing $\eta_1 = 1$ means that there is no early acceptance, and recovers the early rejection weight $w_{\mathrm{er}}$. 
Finally, putting $\eta_1 = \eta_2 = 1$ recovers the original ABC rejection sampling weight $w$.

Using $\eta(\tilde y)$ in \Cref{eq:slweight_contprob} means that $w_{\mathrm{mf}}(\theta)$ can take one of only four possible values:
\begin{equation}
\label{eq:mfcases}
 w_{\mathrm{mf}}(\theta) = \begin{cases}
                   1 &\tilde y \in \tilde \Omega(\tilde \epsilon) \cap U \geq \eta_1 \text{ (early accept)} \\
                   0 &\tilde y \notin \tilde \Omega(\tilde \epsilon) \cap U \geq \eta_2 \text{ (early reject)}\\
                   1 &\tilde y \in \tilde \Omega(\tilde \epsilon) \cap y \in \Omega(\epsilon) \cap U < \eta_1 \text{ (checked true positive)}\\
                   0 &\tilde y \notin \tilde \Omega(\tilde \epsilon) \cap y \notin \Omega(\epsilon) \cap U < \eta_2 \text{ (checked true negative)}\\
                   1-1/\eta_1 &\tilde y \in \tilde \Omega(\tilde \epsilon) \cap y \notin \Omega(\epsilon) \cap U < \eta_1 \text{ (checked false positive)}\\
                   0+1/\eta_2 &\tilde y \notin \tilde \Omega(\tilde \epsilon) \cap y \in \Omega(\epsilon) \cap U < \eta_2 \text{ (checked false negative)}.
                  \end{cases}
\end{equation}
These cases imply the implementation, \cref{a:sl}, of a Monte Carlo algorithm to estimate $\mathbb E_{\mathrm{ABC}}(F(\theta)~|~y_{\mathrm{obs}})$.
They also have the interesting consequence that, in addition to computational speedup, the performance of \cref{a:sl} will be dependent on the Receiver Operating Characteristics (ROC)~\cite{ROC} of the cheap, biased binary classifier $\tilde w(\theta) = \mathbb I(\tilde y \in \tilde \Omega(\tilde \epsilon))$ as an approximation of the expensive binary classifier $w(\theta) = \mathbb I(y \in \Omega(\epsilon))$.

\begin{algorithm}
\caption{Early accept/reject multifidelity ABC}
\label{a:sl}
\begin{algorithmic}
\STATE{Input: observations $y_{\mathrm{obs}}$ and $\tilde y_{\mathrm{obs}}$ from a common experiment; prior $\pi(\cdot)$; function $F(\theta)$; low- and high-fidelity models $\tilde p(\cdot~|~\theta)$ and $p(\cdot~|~\theta)$; distance functions $\tilde d(\cdot,\tilde y_{\mathrm{obs}})$ and $d(\cdot,y_{\mathrm{obs}})$; thresholds $\tilde \epsilon$ and $\epsilon$; continuation probabilities $\eta_1$ and $\eta_2$; Monte Carlo sample size $N$.}
\STATE{~}
\FOR{$i = 1,\dots,N$}
    \STATE{Generate $\theta_i \sim \pi(\cdot)$ and $U \sim \mathrm{Unif}(0,1)$.}
    \STATE{Simulate $\tilde y_i \sim \tilde p(\cdot~|~\theta_i)$ from low-fidelity model.}
    \STATE{Calculate $\tilde w = \mathbb I(\tilde d(\tilde y_i, \tilde y_{\mathrm{obs}})<\tilde \epsilon)$.}
    \STATE{Set $w_i = \tilde w$.}
    \STATE{Set $\eta = \eta_1 \tilde w + \eta_2 (1-\tilde w)$.}
    \IF{$U<\eta$}
        \STATE{Simulate $y_i \sim p(\cdot~|~\theta_i)$ from the high-fidelity model.}
        \STATE{Calculate $w = \mathbb I(d(y_i,y_{\mathrm{obs}})<\epsilon)$.}
        \STATE{Update $w_i = w_i + (w - w_i)/\eta$.}
    \ENDIF
\ENDFOR
\STATE{Calculate $\mu_{\mathrm{ABC}}(F) = \sum_{i=1}^N w_i F(\theta_i) / \sum_{i=1}^N w_i$.}
\RETURN $\{ w_i, \theta_i \}$ and $\mu_{\mathrm{ABC}}(F)$.
\end{algorithmic}
\end{algorithm}

In common with many rejection-sampling approaches, this algorithm is embarrassingly parallel: the for-loop can be implemented across many independent workers.
Furthermore, rejection sampling ABC often relies on a threshold value being specified \emph{a posteriori} to ensure a specific acceptance rate; the distances are ranked and $\epsilon$ is chosen so that the parameter proposals corresponding to the smallest quantile of distances are taken into the sample.
In this setting, we could adapt the algorithm above into two serial components (each of which can still be parallelised).
The first component applies the \emph{a posteriori} thresholding approach to the low-fidelity model alone, giving weights 0 or 1 to each proposed parameter.
In the second, the high-fidelity model is simulated using a random subset of the parameter proposals, chosen based on the continuation probabilities.
The weights are then corrected to give $w_{\mathrm{mf}}$: at this point $\epsilon$ can be chosen to achieve a desired effective sample size (introduced in the next section).
For simplicity, we only consider the case of fixed $\tilde \epsilon$ and $\epsilon$ in the following.

\section{Performance of early accept/reject multifidelity ABC}
\label{ss:comparing}

This section considers the performance of \cref{a:sl} in constructing the Monte Carlo sample $\{ w_{\mathrm{mf}}(\theta_i),\theta_i \}$. 
We discuss how to define the sample quality, and thus how to choose the inputs $(\eta_1,\eta_2)$ to optimise performance.
We will show that the multifidelity approach provides improved performance over rejection sampling ABC, and that early acceptance adds to the benefit of early rejection. 

\subsection{Effective sample size and efficiency}
\label{ss:essefficiency}

Consider a weighted sample $\{ w_i ,\theta_i \}$ output from an importance sampling algorithm. 
The weights $w_i$ correspond to any weighting, for example $w(\theta_i)$ or $w_{\mathrm{mf}}(\theta_i)$. 
We denote the random variable taking values $w_i$ by $W$.
A common measure of the quality of such a sample is its effective sample size (ESS), defined as
\begin{equation}
\label{eq:ESS}
 \mathrm{ESS} = \frac{(\sum_i w_i)^2}{\sum_i w_i^2} = N \frac{(\sum_i w_i/N)^2}{\sum_i w_i^2/N} \approx N \frac{\mathbb E(W)^2}{\mathbb E(W^2)},
\end{equation}
where the approximation is taken in the limit as $N \rightarrow \infty$, and the expectations are across the proposal distribution: in this case, the prior parameter distribution, $\theta \sim \pi(\cdot)$. 
Note that ESS is inversely proportional to a first order approximation of the variance of $\mu_{ABC}(F)$ output by \cref{a:sl}, for any $F$: see the supplementary material \Cref{s:ESS} for more details.
Hence, we will use ESS without requiring $w_i \geq 0$ for all $i$.

\begin{proposition}
Assume one or both of the following holds:
\begin{enumerate}
\item $\eta_1 < 1$ and the false positive probability, $\mathbb P \left( \left\{ \tilde y \in \tilde \Omega(\tilde \epsilon) \right\} \cap \left\{ y \notin \Omega(\epsilon) \right\} \right) > 0$;
\item $\eta_2 < 1$ and the false negative probability $\mathbb P \left( \left\{ \tilde y \notin \tilde \Omega(\tilde \epsilon) \right\} \cap \left\{ y \in \Omega(\epsilon) \right\} \right) > 0$.
\end{enumerate}
In the limit as $N \rightarrow \infty$, the $\mathrm{ESS}$ of the weighted sample $\{w_{\mathrm{mf}}(\theta_i), \theta_i\}$ is smaller than the ESS of the weighted sample $\{w(\theta_i), \theta_i\}$.
\end{proposition}
\begin{proof}
The conditional expectations $\mathbb E(w(\theta)) =\mathbb E(w_{\mathrm{mf}}(\theta)) = p(y \in \Omega(\epsilon)~|~\theta)$ are equal and unbiased. 
It follows that $\mathbb E(w) = \mathbb E(w_{\mathrm{mf}}) = p(y \in \Omega(\epsilon))$, and hence that the numerators of the limiting value of the ESS in \Cref{eq:ESS} are equal for $w_i = w(\theta_i)$ and $w_i = w_{\mathrm{mf}}(\theta_i)$.

It can be shown that $\mathbb E(w^2) = Z = p(y \in \Omega(\epsilon))$. 
Using \Cref{eq:mfcases} and taking expectations, we find
\begin{align}
\mathbb E(w_{\mathrm{mf}}^2) = \mathbb E(w^2) 
    &+ \left( \frac{1}{\eta_1} - 1 \right) \mathbb P \left( \left\{ \tilde y \in \tilde \Omega(\tilde \epsilon) \right\} \cap \left\{ y \notin \Omega(\epsilon) \right\} \right) \nonumber \\
    &+ \left( \frac{1}{\eta_2} - 1 \right) \mathbb P \left( \left\{ \tilde y \notin \tilde \Omega(\tilde \epsilon) \right\} \cap \left\{ y \in \Omega(\epsilon) \right\} \right).
\label{eq:mf_m2}    
\end{align}
Assuming at least one of the two conditions in the statement gives $\mathbb E(w_{\mathrm{mf}}^2) > \mathbb E(w^2)$.
The result follows from this inequality.
\end{proof}

The goal of the multifidelity approach to rejection sampling is to build a sample more efficiently than with standard rejection sampling ABC. 
The smaller ESS produced by \cref{a:sl} is the cost of early acceptance and early rejection.
\Cref{eq:mf_m2} shows that the marginal cost of decreasing either $\eta_1$ or $\eta_2$ is dependent on the probability of either a false positive or false negative, respectively. 
Clearly, if the approximation $\tilde y$ is a good one for $y$ (in terms of the set membership $\tilde y \in \tilde \Omega(\tilde \epsilon)$ as a predictor of $y \in \Omega(\epsilon)$) then the cost of early acceptance or early rejection is reduced. 

Having shown that a smaller ESS is the cost of early acceptance and early rejection, we can now show how this is balanced against the intended benefit of reducing computational burden. 
Suppose that $T_i$ is the time taken to generate the weight $w_i$, with total simulation time $T_{\mathrm{tot}} = \sum_i T_i$.
A measure of the efficiency of building the sample $\{ w_i, \theta_i \}$ is the ratio of ESS to total simulation time,
\begin{equation}
\frac{\mathrm{ESS}}{T_{\mathrm{tot}}} = \frac{(\sum_i w_i/N)^2}{(\sum_i w_i^2/N) (\sum_i T_i/N)}
\approx \frac{\mathbb E(W)^2}{\mathbb E(W^2) \mathbb E(T)},
\label{eq:efficiency_def}
\end{equation}
where we have considered the limit as $N \rightarrow \infty$ and the expectations are taken across $\theta \sim \pi(\cdot)$.

The expected cost $\mathbb E(T)$ of computing $w_{\mathrm{mf}}(\theta)$ over $\theta \sim \pi(\cdot)$ is 
\[
\mathbb E(T) = \mathbb E(\tilde c) + 
\eta_1 \mathbb E(c ~|~ \tilde y \in \tilde \Omega(\tilde \epsilon)) \mathbb P(\tilde y \in \tilde \Omega(\tilde \epsilon)) + 
\eta_2 \mathbb E(c ~|~ \tilde y \notin \tilde \Omega(\tilde \epsilon)) \mathbb P(\tilde y \notin \tilde \Omega(\tilde \epsilon)),
\]
where $\tilde c(\theta)$ is the simulation cost of $\tilde y \sim \tilde p(\cdot~|~\theta)$ and $c(\theta)$ is that of $y \sim p(\cdot~|~\theta)$.
If $\eta_1,\eta_2 < 1 - (\mathbb E(\tilde c)/\mathbb E(c))$, then the expected simulation time $\mathbb E(T)$ to calculate $w_{\mathrm{mf}}$ is less than the expected cost $\mathbb E(c)$ of calculating $w$.
The computational cost of calculating $w_{\mathrm{mf}}(\theta)$ is decreased for smaller values of $\eta_1,\eta_2$, to a lower bound of $\tilde c(\theta)$.
Hence, the benefit of decreasing $\eta_1$ and $\eta_2$ is a saving in computational cost, traded off against a decrease in the ESS. 

\subsection{Optimal continuation probabilities}
\label{s:eta}

\Cref{a:sl} takes the continuation probabilities $(\eta_1,\eta_2)$ as an input, producing a sample $\{ w_{\mathrm{mf}}(\theta_i),\theta_i \}$. 
We now consider the choice of $(\eta_1,\eta_2)$ that optimally balances the benefit of reducing the simulation time against the cost of reducing the ESS.
Our approach is to choose $(\eta_1,\eta_2)$ to maximise the limiting efficiency of the algorithm, defined in \Cref{eq:efficiency_def} as the ratio $\mathrm{ESS}/T_{\mathrm{tot}}$ as $N \rightarrow \infty$.

The numerator, $\mathbb E(w_{\mathrm{mf}})^2 = p(y \in \Omega(\epsilon))^2$, in \Cref{eq:efficiency_def} is independent of $\eta_1$ and $\eta_2$. 
Therefore the efficiency is maximised when the denominator, $\phi(\eta_1,\eta_2) = \mathbb E(w_{\mathrm{mf}}^2) \mathbb E(T)$, is minimised. We define
\begin{subequations}
\label{eq:efficiency}
\begin{align}
 p_{tp} &= \mathbb P \left( \left\{ \tilde y \in \tilde \Omega(\tilde \epsilon) \right\} \cap \left\{ y \in \Omega(\epsilon) \right\} \right),
 \label{eq:p_tp} \\
 p_{fp} &= \mathbb P \left( \left\{ \tilde y \in \tilde \Omega(\tilde \epsilon) \right\} \cap \left\{ y \notin \Omega(\epsilon) \right\} \right), 
  \label{eq:p_fp} \\
 p_{fn} &= \mathbb P \left( \left\{ \tilde y \notin \tilde \Omega(\tilde \epsilon) \right\} \cap \left\{ y \in \Omega(\epsilon) \right\} \right) ,
 \label{eq:p_fn} \\
 c_p &= \mathbb E( c ~|~\tilde y \in \tilde \Omega(\tilde \epsilon)) \mathbb P(\tilde y \in \tilde \Omega(\tilde \epsilon)),
 \label{eq:delta_p} \\
 c_n &= \mathbb E( c ~|~\tilde y \notin \tilde \Omega(\tilde \epsilon)) \mathbb P(\tilde y \notin \tilde \Omega(\tilde \epsilon)),
  \label{eq:delta_n}
\end{align}
\end{subequations}
to write the objective function
\begin{equation}
 \phi(\eta_1,\eta_2) = \biggl( \bigl(p_{tp} - p_{fp}\bigr) + \frac{1}{\eta_1} p_{fp} + \frac{1}{\eta_2} p_{fn} \biggr) \biggl( \mathbb E(\tilde c) + \eta_1 c_p + \eta_2 c_n \biggr).
\label{eq:ESSphi}
\end{equation}
The false positive and false negative probabilities, $p_{fp}$ and $p_{fn}$, respectively, are the average rates at which simulations from the high- and low-fidelity models are different, defined in terms of being close to the data. 
The average computation time, $\mathbb E(c) = c_p + c_n$, to simulate the high-fidelity model is partitioned conditionally on the value of $\mathbb I(\tilde y \in \tilde \Omega(\tilde \epsilon))$.

\begin{lemma}
\label{l:global}
The denominator $\phi(\eta_1,\eta_2)$ has a unique minimiser on $[0, \infty)^2$ if and only if $\mathbb E\left(\left(w - \tilde w\right)^2\right) < \mathbb E(w^2)$. 
The values of $\eta_1, \eta_2 \geq 0$ that minimise $\phi$ over $[0,\infty)^2$ are 
\begin{equation}
\label{eq:global_optimal_eta}
 (\eta_1^\star, \eta_2^\star) = \left( \sqrt{\frac{R_p}{R_0}}, \sqrt{\frac{R_n}{R_0}} \right),
\end{equation}
where
\begin{align*}
 R_p &= \frac{p_{fp}}{c_p / \mathbb E(\tilde c)} 
 = \frac{\mathbb P(y \notin \Omega(\epsilon) ~|~ \tilde y \in \tilde \Omega(\tilde \epsilon))}{\mathbb E(c~|~\tilde y \in \tilde \Omega(\tilde \epsilon)) / \mathbb E(\tilde c)}, \\ 
 R_n &= \frac{p_{fn}}{c_n / \mathbb E(\tilde c)} 
 = \frac{\mathbb P(y \in \Omega(\epsilon) ~|~ \tilde y \notin \tilde \Omega(\tilde \epsilon))}{\mathbb E(c~|~\tilde y \notin \tilde \Omega(\tilde \epsilon)) / \mathbb E(\tilde c) }, \\
 R_0 &= p_{tp} - p_{fp} 
 = \mathbb E(w^2) - \mathbb E\left( \left(\tilde w - w \right)^2 \right).
 \end{align*}
If $R_0 \leq 0$, then $\nabla \phi \neq 0$ globally.
\end{lemma}

\begin{lemma}
\label{l:boundary}
The value of $\eta_1 \in (0,1]$ that minimises $\phi(\eta_1,1)$ is
\[
 \bar \eta_1 = \min \left\{ 1, \eta_1^\star \bigg/ \sqrt{\frac{1 + p_{fn}/R_0}{1 + c_n/\mathbb E(\tilde c)}} \right\}.
\]
The value of $\eta_2 \in (0,1]$ that minimises $\phi(1,\eta_2)$ is
\[
 \bar \eta_2 = \min \left\{ 1, \eta_2^\star \bigg/ \sqrt{\frac{1 + p_{fp}/R_0}{1 + c_p/\mathbb E(\tilde c)}} \right\}.
\]
\end{lemma}

\begin{corollary}
\label{c:optimal-eta}
If $\max \{ R_p, R_n \} \leq R_0$ then the continuation probabilities $\hat \eta_1, \hat \eta_2 \in (0,1]$ that maximise the efficiency, $\mathrm{ESS}/T_{\mathrm{tot}}$, of the sample $\{w_{\mathrm{mf}}(\theta_i), \theta_i \}$ built by \cref{a:sl} are equal to $\eta_1^\star,\eta_2^\star \in (0,1]$ in \Cref{eq:global_optimal_eta}. 
Conversely, if $\max \{ R_p,R_n \} > R_0$ then at least one of $\eta_1^\star, \eta_2^\star > 1$, and the values of $\hat \eta_1,\hat \eta_2 \in (0,1]$ that maximise the efficiency, $\mathrm{ESS}/T_{\mathrm{tot}}$ are:
\[
(\hat \eta_1, \hat \eta_2) = \begin{cases}
\left( 1, \bar \eta_2 \right) & \phi(1,\bar \eta_2) \leq \phi(\bar \eta_1, 1) \\
\left( \bar \eta_1, 1 \right) & \text{else,} 
\end{cases}
\]
where $\bar \eta_1,\bar \eta_2$ are given in \cref{l:boundary}.
\end{corollary}

\begin{proof}
The proofs of \Cref{l:global,l:boundary,c:optimal-eta} are sketched in the supplementary material, \Cref{pf:eta}.
\end{proof}

The optimal continuation probabilities can be interpreted in terms of the ROC analysis of the quality of the low-fidelity classifier $\tilde w(\theta) = \mathbb I(\tilde y \in \tilde \Omega(\tilde \epsilon))$ as an approximation of $w(\theta) = \mathbb I(y \in \Omega(\epsilon))$, and the computational saving of the low-fidelity model over the high-fidelity model.
The false discovery rate, $\mathbb P(y \notin \Omega(\epsilon)~|~\tilde y \in \tilde \Omega(\epsilon))$, and the false omission rate, $\mathbb P(y \in \Omega(\epsilon)~|~\tilde y \notin \tilde \Omega(\epsilon))$, are conditional versions of the false positive and false negative rates in \Cref{eq:p_fp,eq:p_fn}.
The ratio $R_p$ is the false discovery rate, divided by the expected time to simulate the high-fidelity model when $\tilde y \in \tilde \Omega(\tilde \epsilon)$, expressed in units of the low-fidelity simulation cost.
Smaller values of $R_p$ occur when the false discovery rate is small, and where the low-fidelity model is much cheaper than the high-fidelity model.
A similar interpretation exists for $R_n$, defined as the false omission rate divided by the expected time taken to simulate the high-fidelity model when $\tilde y \notin \tilde \Omega(\tilde \epsilon)$, again in units of the low-fidelity simulation cost.

The smallest values of $\eta_1^\star,\eta_2^\star$ are found when $p_{fp}$ and $p_{fn}$ are as small as possible: that is, where the accuracy of $\tilde w$ as an approximation to $w$ is greatest.
As the accuracy decreases, the benefit to the efficiency of putting $\eta_1,\eta_2 < 1$ becomes progressively less, until the optimal choice is for one or both of $\eta_1,\eta_2$ to be unity. 
In such cases, $\tilde w$ is not a good enough approximation to $w$ to recommend any early acceptance and/or rejection at all.

The optimal continuation probabilities make it clear that the early accept/reject multifidelity approach relies on:
(i) the false discovery rate $\mathbb P(y \notin \Omega(\epsilon) ~|~ \tilde y \in \tilde \Omega(\tilde \epsilon))$ and the false omission rate $\mathbb P(y \in \Omega(\epsilon) ~|~ \tilde y \notin \tilde \Omega(\tilde \epsilon))$ being suitably small; and
(ii) the simulation costs $\mathbb E(c~|~\tilde y \in \tilde \Omega(\tilde \epsilon))$ and $\mathbb E(c~|~\tilde y \notin \tilde \Omega(\tilde \epsilon))$ of the high-fidelity model being suitably large in comparison to the average simulation time, $\mathbb E(\tilde c)$, of the low-fidelity model.

\section{Example: stochastic repressilator model}
\label{s:repressilator}

We now illustrate the multifidelity approach to rejection sampling by its application to a stochastic model of a synthetic genetic network known as the repressilator~\cite{Elowitz2000}. 
This synthetic genetic network consists of three genes $G_1$, $G_2$ and $G_3$, which are transcribed and translated into proteins $P_1$, $P_2$ and $P_3$, respectively. 
Transcription of $G_2$ is repressed by $P_1$, transcription of $G_3$ is repressed by $P_2$, and transcription of $G_1$ is repressed by $P_3$. 
This cycle of repression is known to cause oscillatory behaviour.

\subsection{Model}
\label{ss:repressilator}
The specific form of the model is adapted from that used in~\cite{Toni2009}. The chemical reaction description of the model is
\begin{subequations}
\begin{align}
  \xrightarrow{\alpha_0 + \alpha f(p_j)} ~&m_i \xrightarrow{1} \emptyset & &\text{for $(i,j)=(1,3)$, $(2,1)$, and $(3,2)$,} \\
  &m_i \xrightarrow{\beta} m_i + p_i & &\text{for $i = 1,2,3$},  \\
  &p_i \xrightarrow{\beta} \emptyset & &\text{for $i = 1,2,3$},
\end{align}
\label{eq:repressilator}
\end{subequations}
where the decreasing function $f(p) = K_h^n / (K_h^n + p^n) $ models the repression of mRNA transcription by protein.
The goal of parameter identification will be to identify the parameters $n$ and $K_h$. 
For the purposes of this example, the observed data $y_{\mathrm{obs}}$ will be synthetic, generated by simulating the model in \Cref{eq:repressilator} using the `real' parameter values: 
$\alpha_0 =1$, $\beta = 5$, $\alpha = 1000$, $n = 2$, and $K_h = 20$, to a final time of $T_{\mathrm{final}}=10$.
For the parameter inference task, the values of $\alpha_0$, $\alpha$, and $\beta$ are fixed at these nominal values.
The remaining parameters, $n$ and $K_h$, are uncertain with prior distributions
$n \sim U(1, 4)$ and $K_h \sim U(10, 30)$.
The initial conditions are fixed at $(m_1,m_2,m_3) = (0,0,0)$ and $(p_1,p_2,p_3) = (40,20,60)$.

\subsection{Data generation}
\label{sss:data}
We used Gillespie's stochastic simulation algorithm (SSA)~\cite{Gillespie1977} to generate the observed data $y_{\mathrm{obs}}$ for the nominal parameter values. 
Then, for each of $N = 5 \times 10^6$ sample points $(n,K_h)$ from the uniform prior, we generated:
(i) a simulation $\tilde y \sim \tilde p(\cdot~|~\theta)$ from a low-fidelity tau-leap~\cite{Gillespie2001} implementation of \Cref{eq:repressilator}; and 
(ii) a simulation $y \sim p(\cdot~|~\theta)$ from the high-fidelity SSA implementation of \Cref{eq:repressilator}.
For more details of the stochastic simulations, we refer the reader to the tutorial~\cite{Higham2008}, and to \Cref{s:coupling-sims} in the supplementary material.

For each fidelity, the summary statistics are vectors of each species' molecule count at integer time-points $t=0,1,\dots,10$, such that $y_{\mathrm{obs}} = \tilde y_{\mathrm{obs}}$ from the synthetic data.
The distances $\tilde d(\tilde y_{\mathrm{obs}},\tilde y)$ and $d(y_{\mathrm{obs}},y)$ are Euclidean distances normalised by the time horizon, $T_{\mathrm{final}}$, and the threshold values are $\tilde \epsilon = \epsilon = 50$, common to both fidelities.
\Cref{fig:distances} shows, for a subset of $N_{\mathrm{sample}}=10^4$ pairs of simulations, how the values of $d$ vary with $\tilde d$ (left) and $n$ (right).
The left panel shows that the distances from the data of the high- and low-fidelity simulations are correlated.
The quadrants in the left panel also show the correlation between $w$ and $\tilde w$.
We map the false positive and false negative simulations from the left to the right panel, where the orange points show parameter samples where $\tilde y \in \tilde \Omega(\tilde \epsilon)$ but $y \notin \Omega(\epsilon)$, while conversely the green points show parameter values where $\tilde y \notin \tilde \Omega(\tilde \epsilon)$ but $y \in \Omega(\epsilon)$.

\begin{figure}[tbhp]
 \centering
 \includegraphics[width=\textwidth]{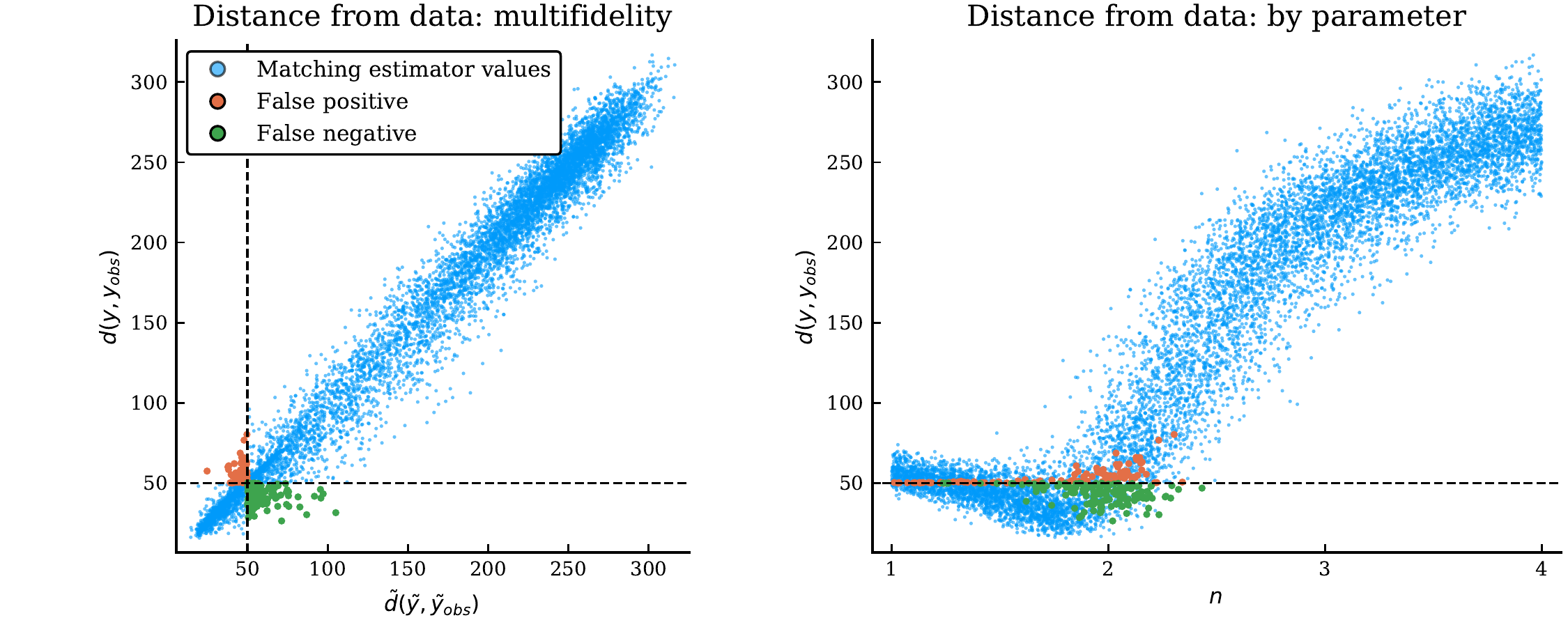}
 \caption{
 Left: Distances between observed data and low-fidelity ($x$-axis) and high-fidelity ($y$-axis) simulations, for $N=10^4$ sample points generated from the uniform prior. 
 Simulations of low- and high-fidelity models are coupled by use of a common noise input (see \Cref{sss:coupling} for details).
 Quadrants correspond to the four possible values of $(\tilde w, w) \in \{ 0, 1 \}^2$.
 Right: Distances between observed data and high-fidelity simulations, plotted against values of $n$ for the same $N_{\mathrm{sample}}=10^4$ sample points of $(n,K_h)$ generated from the uniform prior.
 }
 \label{fig:distances}
\end{figure}

\subsection{Applying early accept/reject multifidelity ABC}
\label{sss:abc}

We use the set of $N = 5 \times 10^6$ simulations as a benchmark dataset,
and assume that the values of the expectations and probabilities in \Cref{eq:efficiency} are given by the empirical expectations and probabilities observed in this dataset.
These values can then be used to calculate the optimal continuation probabilities $(\hat \eta_1, \hat \eta_2) = (0.25, 0.12)$.
In order to demonstrate the optimality of these continuation probabilities, we will compare the efficiency of \cref{a:sl} using the values of $(\eta_1,\eta_2)$ shown in \Cref{fig:ess}. 
We consider: 
\emph{early rejection}, using $\eta_1 = 1$ and $\bar \eta_2 = 0.16$;
\emph{early decision}, using $\eta_1 = \eta_2 = 0.14$;
\emph{rejection sampling}, using $\eta_1 = \eta_2 = 1$; 
and four additional non-optimised values $(\eta_1^\pm,\eta_2^\pm)$, midway between $(\hat \eta_1,\hat \eta_2)$ and each corner of $(0,1]^2$.

\begin{figure}[tbhp]
 \centering
 \includegraphics[width=\textwidth]{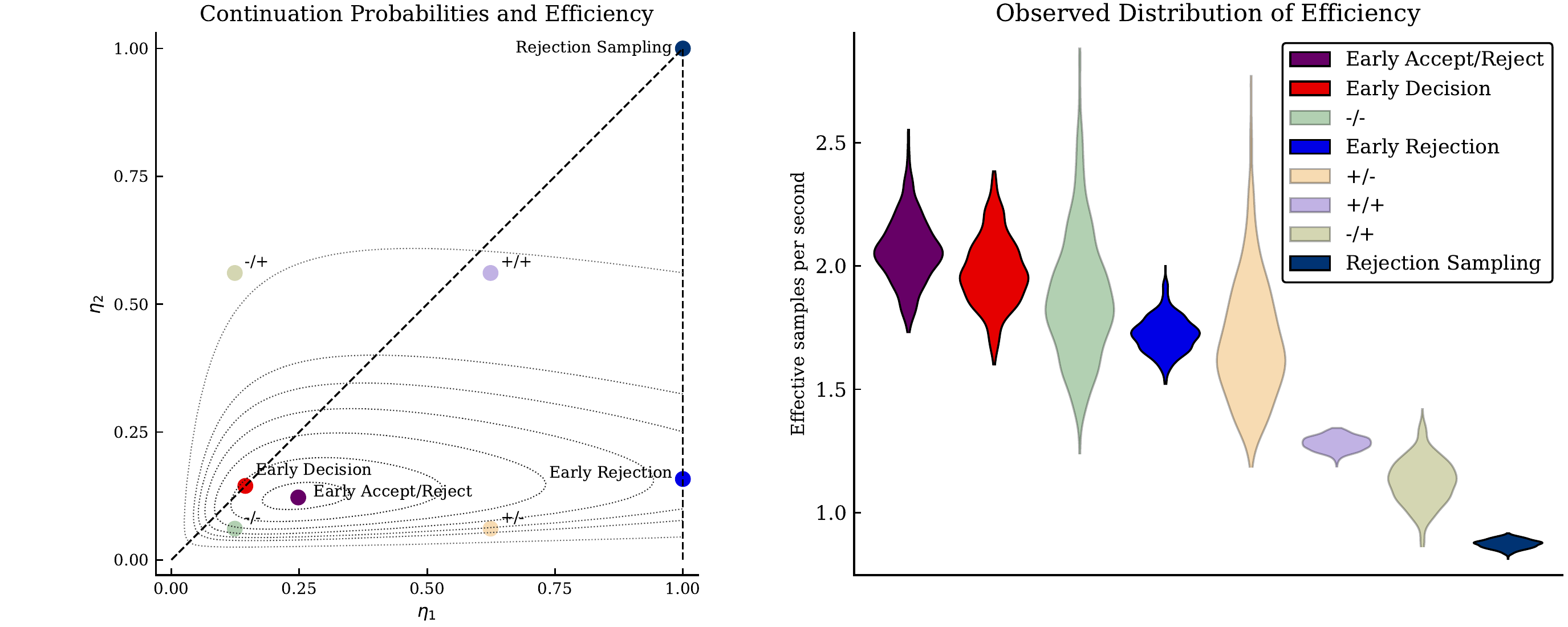}
 \caption{
 Left: values of $(\eta_1, \eta_2)$ used in comparison.
 Level sets of $\phi(\eta_1,\eta_2)$ are depicted, corresponding to $(\eta_1,\eta_2)$ giving 99\%, 95\%, 90\%, 85\%, 80\%, 75\%, and 60\% 
 of the maximum theoretical efficiency $\phi(\hat \eta_1, \hat \eta_2)$.
 The broken vertical and diagonal lines give the constrained spaces in which efficiency is maximised for \emph{early rejection} and \emph{early decision}, respectively.
 Right: observed efficiencies $\mathrm{ESS}/T_{\mathrm{tot}}$ across 500 realisations of \cref{a:sl} for each $(\eta_1,\eta_2)$, ordered by theoretical efficiency $\phi(\eta_1,\eta_2)$. 
 Lighter colours correspond to the non-optimised $(\hat \eta_1^\pm, \hat \eta_2^\pm)$ pairs; darker colours correspond to rejection sampling and the three optimised values of $(\eta_1, \eta_2)$. 
}
 \label{fig:ess}
\end{figure}

\begin{table}[htbp]
{\footnotesize
\caption{The observed probability (across 500 samples built using each $(\eta_1,\eta_2)$) that the efficiency of a realisation using $(\eta_1,\eta_2)$ given by a row exceeds that using $(\eta_1,\eta_2)$ given by a column. The values of $(\eta_1,\eta_2)$ and the distribution of efficiencies are depicted in \cref{fig:ess}.}
\label{tab:pairwise-ess}
\begin{center}
\begin{tabular}{r|ccccccc}
$\mathbb{P}\left(\text{row exceeds column}\right)$ & \text{Early decision} & \text{-/-} & \text{Early rejection} & \text{+/-} & \text{+/+} & \text{-/+} & \text{Rejection}\\ \hline \text{Early accept/reject} & 0.67 & 0.77 & 0.99 & 0.90 & 1.00 & 1.00 & 1.00\\ \text{Early decision} &  & 0.67 & 0.95 & 0.85 & 1.00 & 1.00 & 1.00\\ \text{-/-} &  &  & 0.71 & 0.70 & 1.00 & 1.00 & 1.00\\ \text{Early rejection} &  &  &  & 0.58 & 1.00 & 1.00 & 1.00\\ \text{+/-} &  &  &  &  & 0.99 & 1.00 & 1.00\\ \text{+/+} &  &  &  &  &  & 0.94 & 1.00\\ \text{-/+} &  &  &  &  &  &  & 0.99
\end{tabular}
\end{center}
}
\end{table}

To create the distributions shown in \cref{fig:ess}, we partitioned the benchmark dataset into 500 subsamples of size $N_{\mathrm{sample}}=10^4$.
For each value of $(\eta_1,\eta_2)$, we applied \cref{a:sl} to each of the 500 subsamples and recorded the value of $\mathrm{ESS}/T_{\mathrm{tot}}$\footnote{
Using $(\eta_1,\eta_2) = (1,1)$ in \cref{a:sl} to give the rejection sampling baseline efficiency is slightly unfair: 
\cref{a:abc} is faster, because no low-fidelity simulations are generated. 
However, in \Cref{sss:coupling} we will justify using $(\eta_1,\eta_2) = (1,1)$ in \cref{a:sl} as the rejection sampling baseline.
}.
\Cref{fig:ess} supports the optimality of $(\hat \eta_1,\hat \eta_2) = (0.25, 0.12)$ for maximising the efficiency of \cref{a:sl}.
\Cref{tab:pairwise-ess} quantifies the pairwise comparisons between all eight continuity probability pairs, in terms of how many observed realisations have higher efficiency.
While early decision and early rejection continuation probabilities do improve performance over rejection ABC, the theoretically optimal continuation probabilities $(\hat \eta_1,\hat \eta_2)$ give the highest efficiency across 500 realisations.
For example, we observe that $99\%$ of realisations built using $(\hat \eta_1,\hat \eta_2)$ (early accept/reject) were more efficient than using early rejection alone.
If we enable early acceptance but do not treat the continuation probabilities separately (i.e. early decision), then we still observe that $95\%$ of such realisations are built more efficiently than using early rejection.
However, $67\%$ of samples built treating the early accept/reject continuation probabilities separately are built more efficiently than using a single continuation probability for both.

In summary, we have shown that allowing early acceptance improves performance over early rejection alone.
Furthemore, this benefit is increased by treating early acceptance and early rejection separately, by optimising the continuation probabilities such that $\eta_1 \neq \eta_2$.

\section{Implementation and performance optimisation}
\label{s:implementation}

In this section we discuss the practical issues involved in defining and optimising the performance of \cref{a:sl}, and illustrate them in the context of the example introduced in \Cref{s:repressilator}.

\subsection{Variance reduction by coupling}
\label{sss:coupling}

In \cref{a:sl}, for each $\theta_i \sim \pi(\cdot)$, we first simulate $\tilde y \sim \tilde p(\cdot~|~\theta)$ from the low-fidelity model. 
If $U < \eta(\tilde y)$, we then simulate $y \sim p(\cdot~|~\theta)$ from the high-fidelity model.
In the simplest case, none of the information from the low-fidelity simulation is used to simulate the high-fidelity model: however, this is not optimal.
Consider the specific multifidelity approach of early stopping: in this case, the low-fidelity model replicates the high-fidelity model but only over $[0,t]$, for the stopping time $t<T_{\mathrm{final}}$.
To generate $y$, rather than simulate the model afresh over $[0, T_{\mathrm{final}}]$, we can instead restart the simulation used to generate $\tilde y$ from its state at $t$ and generate the trajectory over $(t,T_{\mathrm{final}}]$.
The high-fidelity model is thus simulated conditional on the low-fidelity simulation.

We can apply this concept to the more general multifidelity setting by simulating the high-fidelity model conditional on the low-fidelity simulation.
Consider a model, $p(\cdot~|~\tilde y, \theta)$, which we will term a \emph{coupling} between the high-fidelity and low-fidelity models, defined such that
\begin{equation}
 \int p(\cdot~|~\tilde y, \theta) \tilde p(\tilde y~|~\theta) ~\mathrm d\tilde y = p(\cdot~|~\theta).
 \label{eq:coupling}
\end{equation}
Given a low-fidelity simulation, $\tilde y \sim \tilde p(\cdot~|~\theta)$, consider a simulation, $y \sim p(\cdot~|~\tilde y, \theta)$, from the coupling, which we will term a \emph{coupled simulation}. 
Then the preceding theory still holds, since \Cref{eq:coupling} implies that the coupled simulation is a simulation of the high-fidelity model, after marginalising out $\tilde y$.

One consequence of the coupled simulation being conditional on the low-fidelity simulation is that $y$ and $\tilde y$, and thus the estimators $w$ and $\tilde w$, will (by a judicious choice of coupling) be correlated.
In the context of \cref{a:sl} this, in turn, means that the false discovery and false omission rates $\mathbb P( y \notin \Omega(\epsilon)~|~ \tilde y \in \tilde \Omega(\tilde \epsilon))$ and $\mathbb P( y \in \Omega(\epsilon)~|~ \tilde y \notin \tilde \Omega(\tilde \epsilon))$, respectively, can be reduced.
This subsequently reduces the variance of $w_{\mathrm{mf}}(\theta)$ as an estimator of the ABC approximation to the likelihood.
A second consequence is that the time $c$ taken to simulate $y \sim p(\cdot~|~\tilde y,\theta)$ from the coupling may be smaller than when simulating $y \sim p(\cdot~|~\theta)$ from the uncoupled high-fidelity model.
If the reuse of information from $\tilde y$ means that the high-fidelity simulation time is smaller, then the optimal rate of early acceptance/rejection is lower\footnote{
Using $(\eta_1,\eta_2) = (1,1)$ in \Cref{a:sl} as the baseline rejection sampler performance can be justified when the time $c$ to simulate $y \sim p(\cdot~|~\theta)$ from the high-fidelity model is equal to the time taken to simulate both $\tilde y$ from the low-fidelity model and the coupled simulation $y \sim p(\cdot~|~\tilde y,\theta)$.
}.

The key problem in this approach is how to define the coupling, $p(\cdot~|~\tilde y,\theta)$.
The appropriate choice of coupling is usually specific to the details of the low- and high-fidelity models~\cite{Giles2015,Lester2015,Peherstorfer2018,Warne2018}. 
The results presented in \Cref{s:repressilator,ss:viral} are based on a coupling between low- and high-fidelity models using a common noise input, as described in \Cref{s:coupling-sims}.


\subsection{Parameter estimation}
\label{ss:Fdep}

Recall that the output of \cref{a:sl} is a set of weights and parameter pairs $\{ w_i, \theta_i \}$ that are used in the estimator
\[
 \mathbb E_{\mathrm{ABC}}(F(\theta)) \approx \frac{1}{N}\sum_i w_i F(\theta_i) \bigg/ \frac{1}{N} \sum_j w_j = \mu_{\mathrm{ABC}}(F).
\]
\Cref{ss:comparing} considered the value of the continuation probabilities to optimise the efficiency $\mathrm{ESS}/T_{\mathrm{tot}}$.
However, the ESS is independent of the function $F$ being estimated by the sample.
We can instead measure the performance of \cref{a:sl} by trading off the variance of the Monte Carlo estimate $\mu_{\mathrm{ABC}}(F)$ against simulation time, a performance metric that is closer to that typically used in multilevel estimation algorithms~\cite{Giles2008}.

\begin{lemma} 
\label{l:varmu}
The variance of $\mu_{\mathrm{ABC}}$ can be expressed in terms of the weights $w_{\mathrm{mf}}(\theta_i)$ as approximately equal to
\[
    \mathbb V \left( \mu_{\mathrm{ABC}} \left( F \right) \right) \approx \frac{1}{NZ^2} \mathbb V \left(w_{\mathrm{mf}} \left( F - \bar F \right) \right) = \frac{1}{N} \frac{\mathbb E \left(w_{\mathrm{mf}}^2 \left( F-\bar F \right)^2 \right)}{\mathbb E\left(w_{\mathrm{mf}}\right)^2},
\]
where $\bar F = \mathbb E_{\mathrm{ABC}}(F~|~y_{\mathrm{obs}})$ is the ABC posterior expectation of $F$ estimated by $\mu_{\mathrm{ABC}}(F)$.
\end{lemma}
\begin{proof}
This expression is derived in the supplementary material, \Cref{s:ESS}.
\end{proof}

\begin{corollary}
 In the limit as $N \rightarrow \infty$, the product $\mathbb V\left( \mu_{\mathrm{ABC}}\left( F \right) \right) T_{\mathrm{tot}}$ of the estimator  variance and the total simulation time can be approximated by
 \begin{equation}
 \label{eq:VarTapprox}
  \mathbb V \left( \mu_{\mathrm{ABC}}\left( F \right) \right) T_{\mathrm{tot}} \approx \frac{\mathbb E\left(w_{\mathrm{mf}}^2 \left( F-\bar F \right)^2 \right) \mathbb E\left( T \right)}{\mathbb E \left(w_{\mathrm{mf}} \right)^2} = \frac{\phi \left(\eta_1,\eta_2; F \right)}{\mathbb E \left(w_{\mathrm{mf}} \right)^2} ,
 \end{equation}
 for the random time $T$ taken to generate $w_{\mathrm{mf}}(\theta)$.
\end{corollary}
Note that the reciprocal of this approximation has a similar form to the limiting value of $\mathrm{ESS}/T_{\mathrm{tot}}$, and can therefore be thought of as an estimator-specific efficiency.
As $\mathbb E(w_{\mathrm{mf}})^2$ is independent of $(\eta_1,\eta_2)$, the optimal tradeoff is where $\phi(\eta_1,\eta_2; F)$ is minimised.

The expected computation time, $\mathbb E(T)$, is given in \Cref{ss:comparing}.
However, the factor corresponding to the second moment is now $F$-dependent, such that
\begin{align*}
 \mathbb E \left(w_{\mathrm{mf}}^2 \left( F - \bar F \right)^2 \right)
 &= 
 \int \left( F(\theta) - \bar F \right)^2 \mathbb E \left(w_{\mathrm{mf}}^2 ~|~ \theta \right) \pi(\theta)~ \mathrm d\theta \\
 &= 
 \big(p_{tp}(F) - p_{fp}(F)\big) + \frac{1}{\eta_1} p_{fp}(F) + \frac{1}{\eta_2} p_{fn}(F),
\end{align*}
where
\begin{subequations}
\label{eq:Fdep}
 \begin{align}
 p_{tp}(F) &= \int \mathbb P \left( \left\{ \tilde y \in \tilde \Omega(\tilde \epsilon) \right\} \cap \left\{ y \in \Omega(\epsilon) \right\} ~|~ \theta \right)  (F(\theta) - \bar F)^2 \pi(\theta) ~\mathrm d\theta, \\
 p_{fp}(F) &= \int \mathbb P \left( \left\{ \tilde y \in \tilde \Omega(\tilde \epsilon) \right\} \cap \left\{ y \notin \Omega(\epsilon) \right\} ~|~ \theta \right)  (F(\theta) - \bar F)^2 \pi(\theta) ~\mathrm d\theta, \\
 p_{fn}(F) &= \int \mathbb P \left( \left\{ \tilde y \notin \tilde \Omega(\tilde \epsilon) \right\} \cap \left\{ y \in \Omega(\epsilon) \right\} ~|~ \theta \right)  (F(\theta) - \bar F)^2 \pi(\theta) ~\mathrm d\theta .
\end{align}
\end{subequations}
In these coefficients, values of $\theta$ generating false positives and false negatives are now penalised based on how much they contribute to the variance.
The optimal continuation probabilities $(\hat \eta_1,\hat \eta_2)$ specific to a given estimator $\mu_{\mathrm{ABC}}(F) \approx \mathbb E_{\mathrm{ABC}}(F(\theta)) = \bar F$ can now be found by replacing $p_{tp}$, $p_{fp}$, $p_{fn}$ in \Cref{l:global,l:boundary,c:optimal-eta} with the respective $F$-dependent parameters in \Cref{eq:Fdep}. 


To illustrate the impact of this alternative performance metric on the continuation probabilities, we return to the repressilator example of \Cref{s:repressilator}. 
We consider three functions of the uncertain parameter $n$ to estimate: $F_1(n) = \mathbb I(n \in (1.9,2.1))$; $F_2(n) = \mathbb I(n \in (1.2,1.4))$; and $F_3(n) = n$.
The optimal pairs $(\hat \eta_1,\hat \eta_2)_i$ for each function are $(\hat \eta_1,\hat \eta_2)_1 = (0.44,0.28)$, $(\hat \eta_1,\hat \eta_2)_2 = (0.23,0.06)$ and $(\hat \eta_1,\hat \eta_2)_3 = (0.38,0.20)$.
These clearly deviate, to different degrees, from the optimal continuation probabilities for maximising $\mathrm{ESS}/T_{\mathrm{tot}}$ of $(\hat \eta_1,\hat \eta_2) =  (0.25, 0.12)$.
Importantly, although the computational time saved by early rejection or acceptance does not change, the contribution of false positives and false negatives to increasing the variance is different enough to change the optimal continuation probabilities. 

To demonstrate the efficiency of each pair of continuation probabilities, we run \Cref{a:sl} on 5000 subsamples of the benchmark data, stopping when the total simulation cost of the subsample reaches 30 seconds.
For each subsample we estimate $\mu(F_i)$: \Cref{tab:varFi} then records the variance across this sample of 5000 estimates.
By fixing the computational cost for every subsample, the variances in the estimator $\mu_{\mathrm{ABC}}(F_i)$ across 5000 subsamples for each value of $(\eta_1,\eta_2)$ can be directly compared.
Using any $(\eta_1,\eta_2) \neq (1,1)$ seems to outperform rejection sampling, but the early accept/reject continuation probabilities are the best-performing. 
By also showing the values of $\phi(\eta_1,\eta_2;F_i)$, we can see that the observed sample variances remain approximately proportional to this objective function.
The percentage values in \Cref{tab:varFi} show that the expected benefits of optimising $\phi(\eta_1,\eta_2; F)$ over $(\eta_1,\eta_2) \in (0,1]^2$ may be marginal, depending on the function, $F$, being estimated.
Furthermore, using $(\eta_1, \eta_2) = (0.25,0.12)$ chosen to optimise ESS clearly does not produce the lowest variances across all functions $F_i$.


\begin{table}[tbhp]
{\footnotesize
 \caption{Continuation probabilities, objective function values (multiplied by $10^3$), and the sample variance (multiplied by $10^3$) of 5000 estimates of $\mu_{\mathrm{ABC}}(F_i)$, $i=1,2,3$, each built with \Cref{a:sl} for a fixed simulation budget of 30 seconds.
 \emph{Optimal ESS} refers to $(\eta_1,\eta_2)$ chosen to minimise ESS, independently of $F$.
 Percentages are reductions relative to rejection sampling (first column).
 }
 \label{tab:varFi}
 \begin{center}
 \begin{tabular}{c|ccccc}
  $\times 10^{-3}$ & Rejection & Early rejection & Early decision & Early accept/reject & Optimal ESS \\
 \hline \hline
 $F_1 = \mathbb I(n \in (1.9,2.1))$ &  \\
 \hline
 $(\eta_1,\eta_2)$ & $(1,1)$ & $(1,0.35)$ & $(0.30,0.30)$ & $(0.44,0.28)$ & $(0.25, 0.12)$ \\
 $\phi(\eta_1,\eta_2;~F_1)$ &  6.14 &  4.68 (24\%) & 4.50 (27\%) & 4.41 (28\%) & 5.09 (17\%) \\
 Sample variance &  4.17  & 3.16 (24\%) & 3.04 (27\%) & 2.85 (32\%) & 3.26 (22\%) \\
 \hline \hline
 $F_2 = \mathbb I(n \in (1.2,2.4))$ &\\
 \hline
 $(\eta_1,\eta_2)$ &  $(1,1)$ & $(1,0.08)$ & $(0.10,0.10)$ & $(0.23,0.06)$ & $(0.25,0.12)$  \\
 $\phi(\eta_1,\eta_2;~F_2)$  &  8.14 & 3.16 (61\%) & 2.85 (65\%) & 2.53 (69\%) & 2.70 (67\%) \\
 Sample variance &  5.63 & 2.03 (64\%) & 1.97 (65\%) &  1.69 (70\%) & 1.82 (68\%) \\
 \hline \hline
 $F_3 = n$ &  \\
 \hline
 $(\eta_1,\eta_2)$ & $(1,1)$ & $(1,0.25)$ & $(0.23,0.23)$ & $(0.38,0.20)$ & $(0.25,0.12)$ \\
 $\phi(\eta_1,\eta_2;~F_3)$ &  4.08 & 2.64 (35\%) & 2.51 (38\%) & 2.41 (41\%) & 2.55 (38\%) \\
 Sample variance &  2.82 & 1.72 (39\%) & 1.68 (40\%) & 1.59 (44\%) & 1.66 (41\%)
 \end{tabular}
 \end{center}
 } 
\end{table}

\subsection{Estimating optimal continuation probabilities}
\label{ss:etaalgo}

The values of $(\hat \eta_1,\hat \eta_2)$ depend on the values given in \Cref{eq:efficiency,eq:Fdep}, which are based on the times taken to generate $\tilde y$ and $y$, together with the ROC values of $\tilde w$ as an approximation of $w$.
Thus, in the absence of any initial information about computation times and ROC values, the optimal continuation probabilities cannot be known in advance. 
Before applying \Cref{a:sl} we therefore need a burn-in period to enable reasonable estimates of the values in \Cref{eq:efficiency,eq:Fdep}. 

Suppose that, at iteration $m$ of \cref{a:sl}, both $\tilde y_i$ and $y_i$ have been generated for $k$ of the $m$ sampled parameter values, $\theta$.
The other $m-k$ values of $\theta$ have only generated $\tilde y$, and have been accepted or rejected early.
We denote the index sets $I_m = \{ 1,\dots,m \}$ and $I_k = \{ i \in I_m \text{ s.t. both }\tilde y_i,y_i \text{ generated} \}$, and write $\rho_m = \sum_{I_m} \mathbb I(\tilde y_i \in \tilde \Omega(\tilde \epsilon))/m$ and $\rho_k = \sum_{I_k} \mathbb I(\tilde y_i \in \tilde \Omega(\tilde \epsilon))/k$.
Natural estimates of the computation times are
\begingroup
\allowdisplaybreaks
\begin{subequations}
\label{eq:time-estimates}
 \begin{align}
  \mathbb E(\tilde c) &\approx \frac{1}{m} \sum_{i \in I_m} \tilde c(\theta_i), \\
  \bar c_p &= \frac{\rho_m}{\rho_k} \cdot \frac{1}{k} \sum_{i \in I_k} c(\theta_i) \mathbb I(\tilde y_i \in \tilde \Omega(\tilde \epsilon)) , \\
  \bar c_n &= \frac{1-\rho_m}{1-\rho_k} \cdot \frac{1}{k} \sum_{i \in I_k} c(\theta_i) \mathbb I(\tilde y_i \notin \tilde \Omega(\tilde \epsilon)).
 \end{align}
\end{subequations}
The remainder of the values in \Cref{eq:efficiency} are similarly estimated by
\begin{subequations}
\label{eq:ROC-estimates}
 \begin{align}
  \bar p_{tp} &= \frac{\rho_m}{\rho_k} \cdot \frac{1}{k} \sum_{i \in I_k} \mathbb I(y_i \in \Omega(\epsilon)) \mathbb I(\tilde y_i \in \tilde \Omega(\tilde \epsilon)), \\
  \bar p_{fp} &= \frac{\rho_m}{\rho_k} \cdot \frac{1}{k} \sum_{i \in I_k} \mathbb I(y_i \notin \Omega(\epsilon)) \mathbb I(\tilde y_i \in \tilde \Omega(\tilde \epsilon)), \\
  \bar p_{fn} &= \frac{1-\rho_m}{1-\rho_k} \cdot \frac{1}{k} \sum_{i \in I_k} \mathbb I(y_i \in \Omega(\epsilon)) \mathbb I(\tilde y_i \notin \tilde \Omega(\tilde \epsilon)) ,
 \end{align}
\end{subequations}
while the $F$-dependent integrals in \Cref{eq:Fdep} are estimated through
\begin{subequations}
\label{eq:Festimates}
 \begin{align}
  \bar p_{tp}(F) &= \frac{\rho_m}{\rho_k} \cdot \frac{1}{k} \sum_{i \in I_k} (F(\theta_i) - \bar \mu)^2 \mathbb I(y_i \in \Omega(\epsilon)) \mathbb I(\tilde y_i \in \tilde \Omega(\tilde \epsilon)) , \\
  \bar p_{fp}(F) &= \frac{\rho_m}{\rho_k} \cdot \frac{1}{k} \sum_{i \in I_k} (F(\theta_i) - \bar \mu)^2 \mathbb I(y_i \notin \Omega(\epsilon)) \mathbb I(\tilde y_i \in \tilde \Omega(\tilde \epsilon)) , \\
  \bar p_{fn}(F) &= \frac{1-\rho_m}{1-\rho_k} \cdot \frac{1}{k} \sum_{i \in I_k} (F(\theta_i) - \bar \mu)^2 \mathbb I(y_i \in \Omega(\epsilon)) \mathbb I(\tilde y_i \notin \Omega(\epsilon)),
 \end{align}
\end{subequations}
\endgroup
where $\bar \mu = \sum_{i \in I_m} F(\theta_i) w_i / \sum_{j \in I_m} w_j$.
In practical implementations of \cref{a:sl}, we propose first beginning with a burn-in run by using $(\eta_1,\eta_2) = (1,1)$ for a suitably large number $M < N$ of sample points $\theta_i$. 
We can then estimate optimal continuation probabilities $(\hat \eta_1,\hat \eta_2)$ using the estimates given in \Cref{eq:time-estimates,eq:ROC-estimates,eq:Festimates} to use for subsequent iterations, $i=M+1,\dots,N$.
Note that the values in \Cref{eq:time-estimates,eq:ROC-estimates,eq:Festimates} will continue to evolve over $i>M$.
We can therefore adapt the continuation probabilities $(\eta_1,\eta_2)$ used for subsequent iterations. 
\Cref{a:sl-adapt} combines a burn-in period of length $M$ with an adaptation of continuation probabilities $(\eta_1,\eta_2)$ towards an evolving estimate of the optimum, subject to lower bounds $\eta_{1,0}$ and $\eta_{2,0}$.

\begin{algorithm}
\caption{Adaptive early accept/reject multifidelity ABC}
\label{a:sl-adapt}
\begin{algorithmic}
\STATE{Input: observations data $y_{\mathrm{obs}}$ and $\tilde y_{\mathrm{obs}}$ from a common experiment; prior $\pi(\cdot)$; function $F(\theta)$; low- and high-fidelity models $\tilde p(\cdot~|~\theta)$ and $p(\cdot~|~\theta)$; distance functions $\tilde d(\cdot,\tilde y_{\mathrm{obs}})$ and $d(\cdot,y_{\mathrm{obs}})$; thresholds $\tilde \epsilon$ and $\epsilon$; lower bounds $\eta_{1,0}$ and $\eta_{2,0}$; Monte Carlo sample size $N$; burn-in length $M<N$.}
\STATE{~}
\STATE{Initialise $(\eta_1,\eta_2) = (1,1)$ and set $I_m = I_k = \emptyset$.}
\FOR{$i = 1,\dots,N$}
    \STATE{Generate $\theta_i \sim \pi(\cdot)$ and $U \sim \mathrm{Unif}(0,1)$.}
    \STATE{Generate $\tilde y_i \sim \tilde p(\cdot~|~\theta_i)$ from the low-fidelity model.}
    \STATE{Calculate $\tilde w = \mathbb I(\tilde d(\tilde y_i, \tilde y_{\mathrm{obs}})<\tilde \epsilon)$.}
    \STATE{Set $w_i = \tilde w$.}
    \STATE{Set $\eta = \eta_1 \tilde w + \eta_2 (1-\tilde w)$.}
    \IF{$U<\eta$}
        \STATE{Generate $y_i \sim p(\cdot~|~\theta_i)$ from the high-fidelity model.}
        \STATE{Calculate $w = \mathbb I(d(y_i,y_{\mathrm{obs}})<\epsilon)$.}
        \STATE{Update $w_i = w_i + (w - w_i)/\eta$.}
        \STATE{Update $I_k = I_k \cup \{i\}$.}
    \ENDIF
    \STATE{Update $I_n = I_n \cup \{i\}$.}
    \STATE{Set $n = |I_n|$ and $k = |I_k|$.}
    \IF{$k \geq M$}
        \STATE{Recalculate values in \Cref{eq:time-estimates,eq:ROC-estimates,eq:Festimates}.}
        \STATE{Estimate optimal $\hat \eta_1$ and $\hat \eta_2$ according to \Cref{l:global,l:boundary,c:optimal-eta}.}
        \STATE{Update $\eta_1 = \hat \eta_1$ and $\eta_2 = \hat \eta_2$.}
        \STATE{Update $\eta_1 = \min(\eta_1, \eta_{1,0})$ and $\eta_2 = \min(\eta_2, \eta_{2,0})$.}
    \ENDIF
\ENDFOR
\STATE{Calculate $\mu_{\mathrm{ABC}} = \sum_{i=1}^N w_i F(\theta_i) / \sum_{i=1}^N w_i$.}
\RETURN $\mu_{\mathrm{ABC}}$
\end{algorithmic}
\end{algorithm}

This algorithm is no longer `embarassingly' parallel, although many copies of the for loop could run independently to produce a sample, potentially exchanging information on an optimal $(\eta_1, \eta_2)$.
It also requires \emph{a priori} fixed $\epsilon$ and $\tilde \epsilon$ for the optimal continuation probabilities to be well-defined, and therefore cannot target a particular acceptance rate.
More importantly, there are no longer guarantees of the consistency of the resulting estimate, although the example in the following section shows good performance.
To guarantee consistency, the adaptive phase may be followed by running \Cref{a:sl} with fixed continuation probabilities equal to $(\eta_1,\eta_2)$ found by the end of the adaptive phase.

\section{Example: viral kinetics}
\label{ss:viral}

We conclude with a further example using a model of intracellular viral kinetics~\cite{Haseltine2002,Srivastava2002}.

\subsection{Model}
A cell is initially infected with a single viral template. 
Templates hijack cellular processes to produce new viral genomes and structural protein, which combine to produce new viral vectors that are expelled from the cell. 
Alternatively, viral genomes can become new templates, and templates and structural protein can also decay. 
We denote the counts of each molecule at time $t$ by the vector 
\[
(\mathrm{template}(t), \mathrm{genome}(t), \mathrm{struct}(t), \mathrm{virus}(t)) = (x_1(t),x_2(t),x_3(t),x_4(t)).
\]
The six reactions are
\begingroup
\allowdisplaybreaks
\begin{subequations}
\begin{align}
 \mathrm{template} &\xrightarrow{k_1} \mathrm{template} + \mathrm{genome}, \\
 \mathrm{genome} &\xrightarrow{k_2} \mathrm{template}, \\
 \mathrm{template} &\xrightarrow{k_3} \mathrm{template} + \mathrm{struct}, \label{eq:fast1} \\
 \mathrm{template} &\xrightarrow{k_4} \emptyset,\\
 \mathrm{struct} &\xrightarrow{k_5} \emptyset, \label{eq:fast2} \\
 \mathrm{genome} + \mathrm{struct} &\xrightarrow{k_6} \mathrm{virus}.
\end{align}
\label{eq:viral-kinetics}
\end{subequations}
\endgroup
We use initial conditions $x(0) = (1,0,0,0)$ and time horizon $[0, T_{\mathrm{final}}] = [0, 200]$.
One important characteristic of this system is that cells can randomly recover from small-scale infection, whenever the template decays before enough genome is produced to set off the positive feedback loop leading to viral infection. 
Even if not recovered, cells can stay latently infected for a randomly long period of time before $\mathrm{virus}(t)>0$.

\subsection{Data generation}
\label{sss:viral-data}

The goal of parameter identification will be to identify the reaction rates $k_i$, $i=1,\dots,6$ in \Cref{eq:viral-kinetics}. 
We first generate synthetic data $y_{\mathrm{obs}}$, using the exact Gillespie SSA~\cite{Gillespie1977} with nominal parameters $(1,0.025,100,0.25,1.9985,7.5\times 10^{-5})$.
Ten independent simulations are produced, each corresponding to a cell in a population of size ten with a common, randomly selected, parameter set. 
The prior distribution on each uncertain parameter $k_i$ is log-uniform around its nominal value; that is, we multiply the nominal value of $k_i$ by $1.5^{u_i}$ for $u_i \in U(-1,1)$. 
The initial conditions are fixed at a single viral template, $x_1(0)=1$.

The low-fidelity model is an adaptation of that given in~\cite{Haseltine2002}. 
For the parameter ranges considered in this example, the propensities of the reactions in \Cref{eq:fast1,eq:fast2} are extremely large relative to those of the other reactions in \Cref{eq:viral-kinetics}. 
Low-fidelity model simulations are therefore generated using a hybrid stochastic/deterministic algorithm~\cite{Haseltine2002} that avoids the computational bottleneck arising from excessive firings of the fast reactions.
We approximate these reactions by considering only their net effect on the mean molecule count, ignoring the fast stochastic fluctuations around the slowly-evolving mean.
In this example, we simulate the high-fidelity model conditional on the simulation of the low-fidelity model using a coupling, $p(\cdot~|~\tilde y,\theta)$, that shares the random noise input between the two simulations.
For more details on the coupling approach, see the supplementary material, \Cref{s:coupling-sims}.

For $M = 10^5$ sample parameters generated from the prior distribution for $(k_1,\dots,k_6)$, we produced ten simulations from the low-fidelity model with ten coupled simulations from the high-fidelity model, corresponding to populations of size ten cells for each parameter vector.
The summary statistics are defined as follows. 
First, a cell is considered infected if it has output a nontrivial number of virus replicates over the 200-day horizon, such that $x_4(200) > 3$. 
Each population thus has a number of infected cells: the three-dimensional summary statistics $y$ and $\tilde y$ are 
(i) the infected percentage of the population,
(ii) $\log_2$ of the average viral output of each infected cell by $t=200$, and
(iii) the average percentage along the time horizon that an infected cell first exceeds the detection threshold of $3$.
If there are zero infected cells, we use the zero vector.
The distances $\tilde d(\cdot,\tilde y_{\mathrm{obs}})$ and $d(\cdot,y_{\mathrm{obs}})$ are both the Euclidean distance between summary statistics, shown in \cref{fig:viral-dist} for $N=10^4$ pairs of simulations.

\begin{figure}[tbhp]
 \centering
 \includegraphics[width=0.7\textwidth]{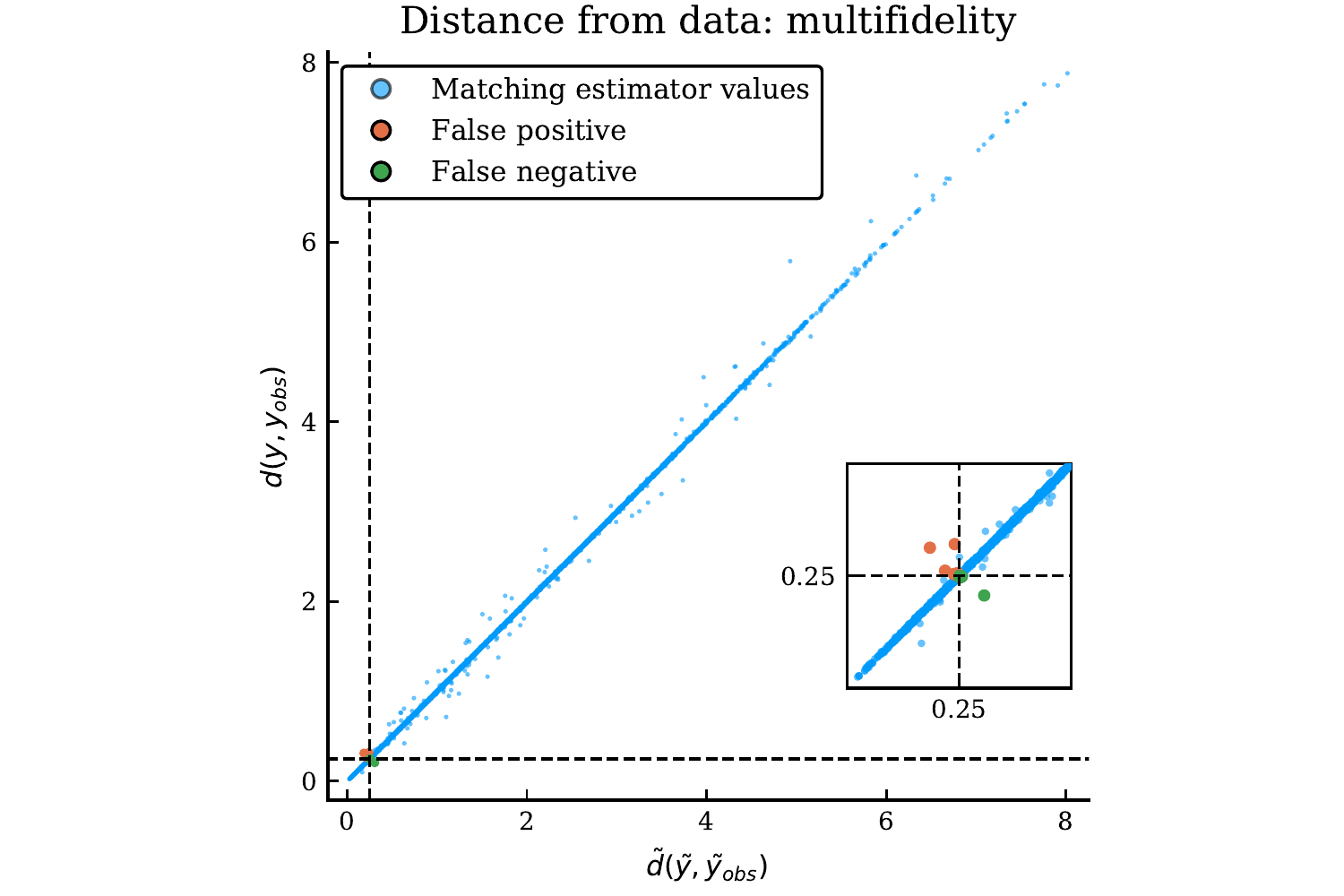}
 \caption{
 Distances between observed data and low-fidelity ($x$-axis) and coupled high-fidelity ($y$-axis) simulations, for $M=10^4$ sample points generated from the six-dimensional log-uniform prior. 
 Quadrants correspond to the four possible values of $(\tilde w, w) \in \{ 0, 1 \}^2$.
 Inset: Region where $\tilde d$ and $d$ close to $\tilde \epsilon = \epsilon = 0.25$.}
 \label{fig:viral-dist}
\end{figure}

Note that, in comparison to the repressilator example (\cref{fig:distances}), the distances in this case are much more closely correlated.
However, the relative speed-up in simulation times is not significantly different: the average cost of a low-fidelity simulation is 17.0\% of an average high-fidelity simulation in the repressilator example, compared to 17.4\% in this example.
The improved accuracy in \cref{fig:viral-dist} suggests that the optimal continuation probabilities $(\hat \eta_1,\hat \eta_2)$ should be smaller, as fewer corrections will be needed.

\subsection{Applying early accept/reject multifidelity ABC}
\label{sss:viral-mlABC} 

We return to measuring a sample's quality by $\mathrm{ESS}/T_{\mathrm{tot}}$.
Taking the full set of $10^5$ pairs of simulations implies optimal continuation probabilities of $(\hat \eta_1,\hat \eta_2) = (0.161, 0.048)$.
We produced 100 independent runs of the adaptive phase of \Cref{a:sl-adapt}, using this full set as the burn-in set each time: thus, the adaptive $(\eta_1,\eta_2)$ values began at $(0.161,0.048)$.
The left-most plot in \Cref{fig:adaptive-efficiency} shows the observed distribution of the efficiencies of the 100 samples built during the adaptive phase.
This is clearly multimodal: some samples are built much less efficiently than others.

This is a consequence of $(\hat \eta_1, \hat \eta_2)$ being the optimal continuation probabilities only in the asymptotic limit.
Due to the accuracy of the low-fidelity model, shown in \Cref{fig:viral-dist}, observed misclassifications $w(\theta_i) \neq \tilde w(\theta_i)$ are relatively rare events within a finite sample.
When these rare events do happen, they lead to a much smaller ESS.
For example, assuming that the continuation probabilities stay approximately equal to $(\hat \eta_1, \hat \eta_2)$, then $w_{\mathrm{mf}} = {-5.22}$ for a false positive and $w_{\mathrm{mf}} = 20.82$ for a false negative.
Each realisation used to construct the left-most plot in \Cref{fig:adaptive-efficiency} thus effectively contains a Poisson number of weights $w_{\mathrm{mf}}(\theta_i) \in \{ -5.22, 20.82 \}$, each of which significantly decreases the ESS, inducing a multimodal distribution for efficiency.



\begin{figure}[tbhp]
 \centering
 \includegraphics[width=0.65\textwidth]{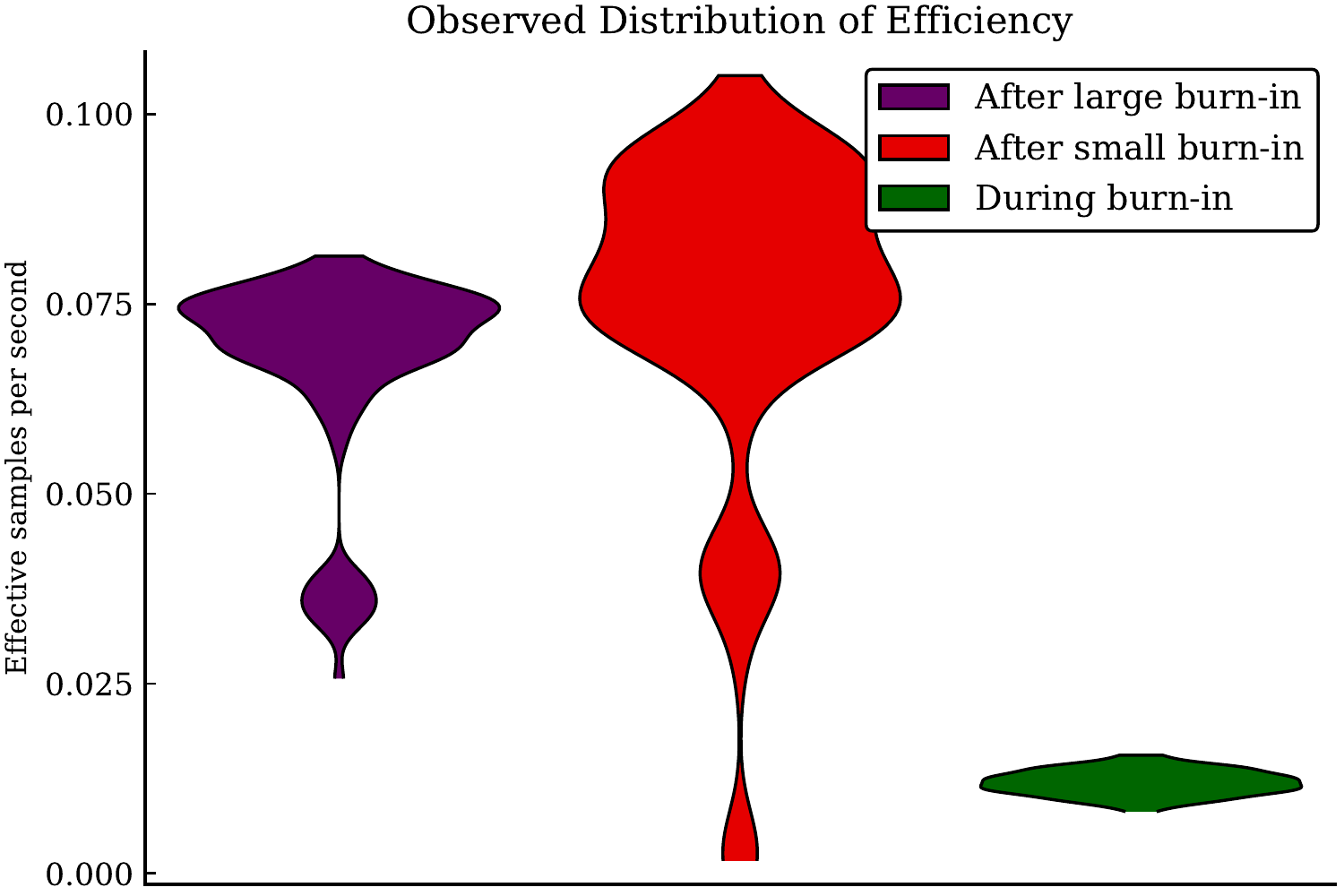}
 \caption{
The efficiency distribution of 100 samples built using \Cref{a:sl-adapt}.
`After large/small burn-in' depicts the efficiencies of the part of the samples built during the adaptive phase, with a starting value of $(\eta_1,\eta_2)$ derived using a common burn-in sample of size $10^5$ (large) or 100 independent burn-in samples of size $10^3$ (small).
`During burn-in' depicts the distribution of the efficiency for each of the 100 small burn-in phases.
}
 \label{fig:adaptive-efficiency}
\end{figure}

Recall that the objective functions $\phi(\eta_1,\eta_2)$ and $\phi(\eta_1,\eta_2;F)$ are the products of the limiting values of the second moment and computation time as the sample size, $N$, approaches infinity. 
This example demonstrates that when the false discovery rate and false omission rate, and hence the continuation probabilities, are particularly small, the effect of finite $N$ becomes more important.
We hypothesise that the sample size $N$ can be considered large enough for an accurate estimate of $p_{fp}$ etc. (and hence of the optimal continuation probabilities) only when the number of weights $ w_{\mathrm{mf}}(\theta_i) \notin \{ 0,1 \} $ is suitably large.
Future work could potentially aim to further optimise the continuation probabilities by taking into account a fixed $N$ or computational budget $\sum_i T_i < B$ more explicitly.

\subsection{Shorter burn-in estimates}
\label{sss:burn-in}

The burn-in set of $10^5$ pairs of simulations took 132 hours of computation time.
We created a further 100 independent samples using the adaptive phase of \Cref{a:sl-adapt}, but this time also partitioned the burn-in set into 100 independent subsamples of size $M=10^3$.
The centre and right plots in \Cref{fig:adaptive-efficiency} show the distributions of efficiency across the 100 samples using this shorter burn-in phase, during the adaptive phase (red) and initial burn-in (green). 
Clearly, the portions of each sample built during the burn-in phase are much built less efficiently, on average, than the portions of the samples built during the adaptive phase.
However, the small burn-in duration leads to an even more pronounced multi-modal efficiency distribution during the adaptive phase, and the effective sample size of some samples has collapsed due to large-magnitude weights.

To observe how far the weights are from the optimum, \Cref{fig:eta-cloud} shows the variability in the continuation probabilities used when applying \cref{a:sl-adapt}.
The values of $(\eta_1,\eta_2)$ used at the beginning of the adaptive phase are shown in blue, and in orange are the resulting values at the end of the adaptive phase.
A point lies at each of the fixed lower boundaries $\eta_1 = 0.01$ or $\eta_2 = 0.01$ if no false positive or no false negative has been observed, respectively.
The continuation probabilities may lie on a lower boundary at the start of the adaptive phase, but during the adaptive phase a false positive or false negative may be observed.
The resulting sample will then include a weight of $100$ or ${-99}$: these are the samples of extremely low efficiency show in the red plot in \Cref{fig:adaptive-efficiency}, as the effective sample size will be significantly decreased.
However, the continuation probabilities that lie on a lower boundary at the end of the adaptive phase (i.e. the orange points) are those where no false positive or false negative has been observed during either the burn-in or the adaptive phase.
These are the samples with extremely high efficiency in the red plot in \Cref{fig:adaptive-efficiency}.
Similarly to the case of a long burn-in phase, future development of the adaptive approach should focus on identifying corrections to $(\eta_1,\eta_2)$ to account for these finite sample size effects.

\begin{figure}[tbhp]
 \centering
 \includegraphics[width=0.65\textwidth]{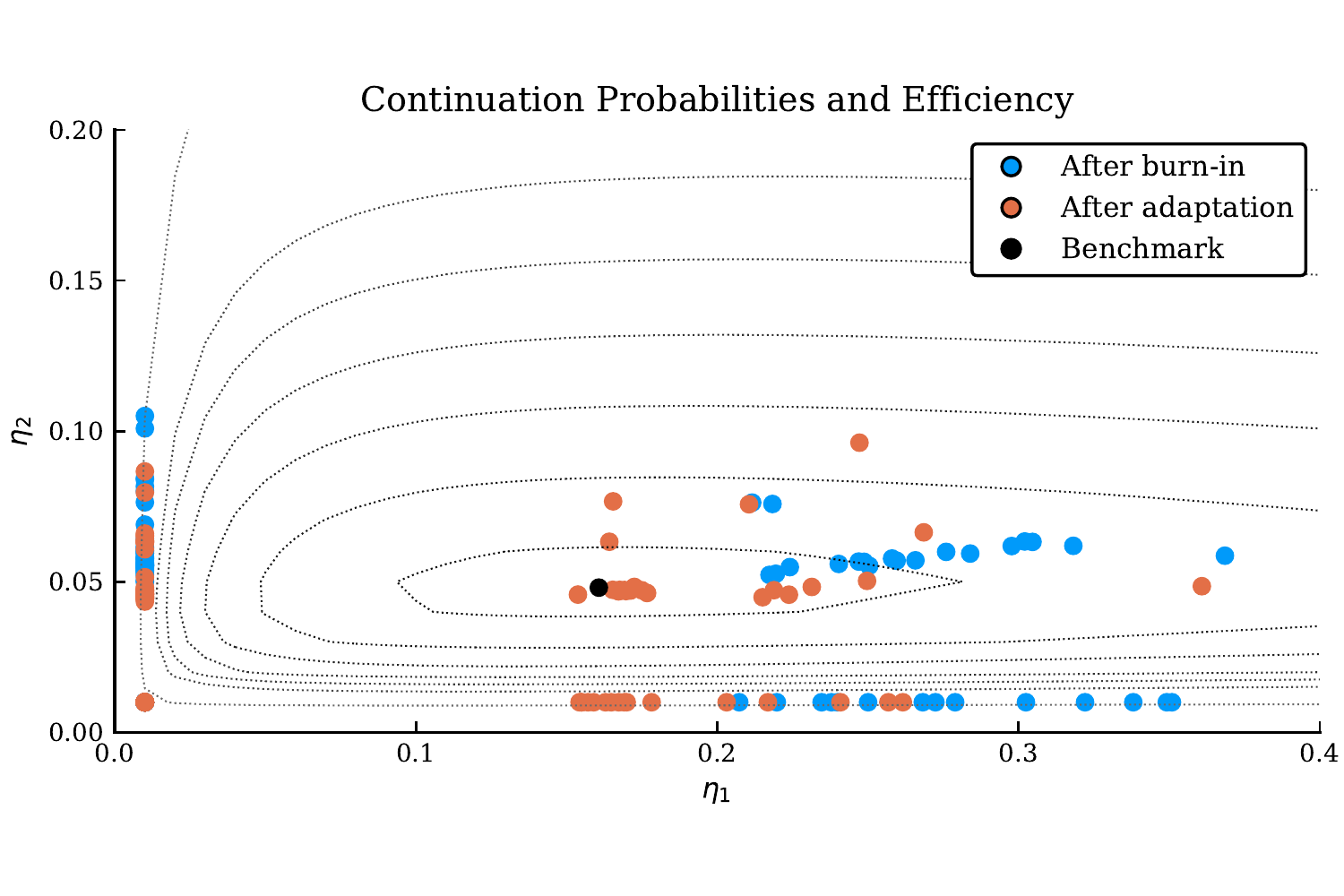}
 \caption{Cloud of near-optimal estimates of $(\hat \eta_1,\hat \eta_2)$ output by \cref{a:sl-adapt}. The black point is the `true' value of $(\hat \eta_1,\hat \eta_2)$ estimated using the entire benchmark dataset. Contours are level sets corresponding to $99\%$, $95\%$, $90\%$, $85\%$, $80\%$, $75\%$ and $60\%$ of the theoretical maximum efficiency achieved at $(\hat \eta_1,\hat \eta_2)$.}
 \label{fig:eta-cloud}
\end{figure}


\section{Discussion and conclusions}
\label{s:end}

In this work, we have considered the use of multifidelity methods to improve the efficiency of constructing ABC estimators by optimally combining high- and low-fidelity models.
We combined the strengths of early rejection and early decision approaches to construct a multifidelity method with both early acceptance and early rejection of parameter samples, which treats the choice whether to simulate the high-fidelity model differently, depending on the output of the low-fidelity simulation.
One consequence of this is that parameter samples for which the high-fidelity model is simulated are not distributed across the parameter space according to the prior, $\pi$.
The early accept/reject method can thus be interpreted as an importance sampling approach, with an importance distribution induced by the low-fidelity model.

The samples built in \cref{a:sl,a:sl-adapt} will contain negative weights whenever $\tilde y \in \tilde \Omega(\tilde \epsilon)$ and $y \notin \Omega(\epsilon)$. 
These negative weights means that the resulting set of weights and parameters $\{ w_i, \theta_i \}$ cannot be treated as a weighted sample from the ABC posterior. 
As a result, \Cref{a:sl} with $\eta_1<1$ cannot currently be adapted to methods reliant on resampling, such as SMC-ABC~\cite{Marjoram2003,Sisson2018,Toni2009}, or to the rejection approach of discarding the proposal used in MCMC-ABC~\cite{Wiqvist2018}.
An SMC approach will allow the acceptance thresholds $\epsilon$ and $\tilde \epsilon$ to be selected adaptively rather than be fixed \emph{a priori}, and for the continuation probabilities to adapt with them.
Therefore, future work should look to apply multilevel approaches to directly building samples from the ABC posterior.

Up to now, we have considered only a single low-fidelity model. 
There are often situations where there are multiple competing low-fidelity models, the accuracy and computational savings of which varies across parameter space.
The low-fidelity models therefore do not necessarily form a hierarchy of progressively increasing accuracy or cost that is valid across all of parameter space, although such hierarchies may exist locally~\cite{Peherstorfer2016}.
For example, both the accuracy and relative speed-up of the approximation in \Cref{ss:viral} will vary with parameters $k_3$ and $k_5$.
One strength of the multifidelity method proposed here is that the requirement for high-fidelity simulations varies across parameter and simulation space, without any analytical input.
Hence, we expect that adapting the approach described in this work to situations with multiple low-fidelity models should focus computational effort towards the models that give the greatest benefits, potentially uncovering local hierarchies in model fidelity in the process.

The continuation probability $\eta(\tilde y)$ was chosen in \Cref{eq:slweight_contprob} to depend on the value of $\mathbb I(\tilde y \in \tilde \Omega(\tilde \epsilon))$, in order to implement early acceptance and early rejection with constant probabilities. 
However, there is no reason to constrain $\eta(\tilde y)$ to this form.
Further work in this area could also explore the potential of generalisations, such as 
$
\eta(\tilde y) = \sum_{r=1}^R \eta_r \mathbb I(\tilde y \in \Omega_r)
$
, for any partition $\Omega_r$ of the output space of the low-fidelity model, $\tilde p(\cdot~|~\theta)$. 
For example, we could consider
$
\Omega_r = \{ \tilde d(\tilde y, \tilde y_{\mathrm{obs}}) \in (\tilde \epsilon_r, \tilde \epsilon_{r-1}) \}
$,
for a decreasing sequence of thresholds $\{ \tilde \epsilon_r \}$. 
Another option would be to also include explicit $\theta$ dependence into the continuation probability $\eta(\tilde y,\theta)$ to reflect, for example, the effect of $\theta$ on the times taken to simulate $\tilde y \sim \tilde p(\cdot~|~\theta)$ and $y \sim p(\cdot~|~\theta)$, or knowledge about $F(\theta)$.

\Cref{ss:etaalgo} discusses one way of dealing with the lack of \emph{a priori} knowledge on the ROC analysis of the cheap rejection sampler $\tilde w = \mathbb I(\tilde y \in \tilde \Omega(\tilde \epsilon))$ as an approximation to the expensive rejection sampler $w = \mathbb I(y \in \Omega(\epsilon))$ and hence of the optimal continuation probabilities.
However, different application areas may provide low-fidelity models with known error bounds relative to the high-fidelity models, such as standard results on balanced truncation~\cite{Gugercin2004} for example. 
It may be possible to use these bounds to reduce uncertainty in the ROC values more efficiently than in \cref{a:sl-adapt}. 
This approach is likely to be much more application-driven, as much error estimation theory for model reduction is based on specific model reductions and specific model outputs and summary statistics~\cite{Benner2017,Prescott2012}.

\section*{Acknowledgments}
Ruth E.\ Baker is a Royal Society Wolfson Research Merit Award holder and a Leverhulme Research Fellow. Thomas P.~Prescott and Prof. Baker also acknowledge the Biotechnology and Biological Sciences Research Council for funding via grant no. BB/R000816/1. 

\bibliographystyle{siamplain}
\bibliography{../../../refs/library}

\newpage
\appendix

\section{Link between ESS and variance}
\label{s:ESS}

The justification for using the effective sample size (ESS)~\cite{Kong1992} is based on estimating the variance of the estimator of $\mathbb E_{\mathrm{ABC}}(F(\theta)~|~y_{\mathrm{obs}})$ given by
\[
 \mu_{\mathrm{ABC}} = \frac{\sum w_i F(\theta_i) / N}{\sum w_i / N}.
\]
Using the delta method~\cite{Seltman2018, Kendall2010}, the variance of this estimator is approximated as
\[
 \mathbb V(\mu_{\mathrm{ABC}}) 
 \approx 
 \frac{1}{N \mu_W^2} \left( \left( \frac{\mu_{WF}}{\mu_W} \right)^2 \mathbb V(W) + \mathbb V( WF ) - 2\left( \frac{\mu_{WF}}{\mu_W} \right) \mathrm{Cov}(W, WF) \right),
\]
where $W = w_i$ and $WF = w_i F(\theta_i)$ are the random sample weights and weighted sample values, respectively, and we denote the expectations as $\mu_{WF} = \mathbb E(\sum w_i F(\theta_i) / N) = \mathbb E(WF)$ and $\mu_W =  \mathbb E(\sum w_i / N) = \mathbb E(W)$, respectively.

Since $w_i = w_{\mathrm{mf}}(\theta_i)$ is unbiased, the expectations in this expression can be written as
\begin{align*}
 \mu_{WF} &= Z \mathbb E_{\mathrm{ABC}}(F(\theta)~|~y_{\mathrm{obs}}), \\
 \mu_{W} &= Z,
\end{align*}
where $Z = \mathbb P(y \in \Omega(\epsilon))$ is the normalisation constant for the ABC posterior in \Cref{eq:abcposterior}. 
Writing $\bar F = \mathbb E_{\mathrm{ABC}} (F(\theta)~|~y_{\mathrm{obs}})$, the approximation to the variance is therefore equal to 
\begin{align}
 \mathbb V(\mu_{\mathrm{ABC}}) &\approx 
 \frac{1}{N \mathbb E(W)^2} \left( \bar F^2 \mathbb V(W ) + \mathbb V(WF) - 2 \bar F \mathrm{Cov} (W, WF) \right) \nonumber \\
 &= \frac{\mathbb V \left( W(F-\bar F) \right)}{N\mathbb E(W)^2} = \frac{\mathbb E \left( W^2 (F-\bar F)^2 \right)}{N\mathbb E(W)^2}, \label{eq:varapprox}
\end{align}
where the final equality follows from $\mathbb E(WF) = \bar F \mathbb E(W)$. This derivation leads to the expression in \Cref{l:varmu}, approximating the variance of the estimator $\mu_{\mathrm{ABC}}$ of $\bar F$.

The ESS derivation in \cite{Kong1992} makes a further approximation to remove dependence on $F$.
Following this, we further approximate the variance as
\[
 \mathbb V(\mu_{\mathrm{ABC}}) \approx \frac{\alpha \mathbb E(W^2)}{N \mathbb E(W)^2},
\]
for a constant $\alpha$.
We repeat the caveat from~\cite{Kong1992} that this approximation can be substantially inaccurate, but also repeat that this approximation usefully means that $\mathbb V(\mu_{\mathrm{ABC}})$ can be written independently of $F$.

This uncovers the link between the variance of $\mu_{\mathrm{ABC}}$ and the ESS. From the definition of the ESS in \Cref{eq:ESS} we find that
\[
 \mathrm{ESS} = \frac{(\sum w_i)^2}{\sum w_i^2} \approx N \frac{\mathbb E(W)^2}{\mathbb E(W^2)} \approx \frac{\alpha}{\mathbb V(\mu_{\mathrm{ABC}})}.
\]
Hence, even in the case where the weights may take negative values, $w_i<0$, the ESS remains a reasonable heuristic for quantifying the quality of the size of the sample: to a rough approximation, the decay of the estimator variance is inversely proportional to the ESS.
Note that, in \Cref{ss:Fdep}, we avoid the second approximation in this derivation and instead use the more accurate approximation of the variance of the estimate (given in \Cref{eq:varapprox}) to quantify the sample quality.

\section{Optimising efficiency}
\label{pf:eta}

The function $\phi(\eta_1,\eta_2)$ in \Cref{eq:ESSphi} is equal to
\[
\phi(\eta_1,\eta_2) 
=
\biggl[ \bigl(p_{tp} - p_{fp} \bigr) + \frac{1}{\eta_1} p_{fp} + \frac{1}{\eta_2} p_{fn}  \biggr] \biggl[ \mathbb E(\tilde c) + \eta_1 c_p + \eta_2 c_n \biggr].
\]
We aim to minimise this function on $\eta_i \in (0,1]^2$.

\begin{lemma}
Let $Y,a,b,c,d>0$ be positive constants in the function $\phi(x_1,x_2) = (X+a/x_1 + b/x_2) (Y+cx_1+dx_2)$, defined on the positive quadrant $x_1,x_2 > 0$. If $X>0$ then this function has a single minimum at $x_1^\star = \sqrt{aY/cX}$ and $x_2^\star = \sqrt{bY/dX}$.
If $X \leq 0$ then there is no minimiser.
\end{lemma}
\begin{proof}
First, if $X<0$ then $\phi \rightarrow -\infty$ is unbounded as both $x_1,x_2 \rightarrow \infty$. Expanding $\phi$ gives
\[
\phi(x_1,x_2) = XY + ac + bd + \left( \frac{aY}{x_1} + cXx_1 \right) + \left( \frac{bY}{x_2} + dXx_2 \right) + \left( bc \frac{x_1}{x_2} + ad \frac{x_2}{x_1} \right)
\]
When $X=0$ then, for any ratio $x_1/x_2 = \lambda>0$, the value of $\phi$ will decrease as $x_1 = \lambda x_2 \rightarrow \infty$. Thus there can be no minimiser in the positive quadrant for $X \leq 0$.

This leaves the case where $X>0$. Writing $z_1 = \sqrt{cX/aY} x_1$ and $z_2 = \sqrt{dX/bY} x_2$ gives
\begin{align*}
\phi(z_1,z_2) &= XY + ac + bd \\ 
&\quad + \sqrt{acXY} \left( z_1 + \frac{1}{z_1} \right)
+ \sqrt{bdXY} \left( z_2 + \frac{1}{z_2} \right)
+ \sqrt{acbd} \left( \frac{z_1}{z_2} + \frac{z_2}{z_1} \right) \\
&= XY + ac + bd \\
&\quad + 2\left[ \sqrt{acXY} \cosh(\ln(z_1)) + \sqrt{bdXY} \cosh(\ln(z_2)) + \sqrt{acbd} \cosh(\ln(z_1)- \ln(z_2)) \right].
\end{align*}
This function clearly has a unique minimum when $z_1 = z_2 = 1$; that is, for $x_i = x_i^\star$.
\end{proof}

\begin{corollary}
If $aY/cX \leq 1$ and $bY/dX \leq 1$ then $\phi$ is minimised on $(0,1]$ at $x^\star$. Else $\phi$ is minimised on the set $\{ x ~|~ \max(x_1,x_2) = 1\}$ forming the boundary of $(0,1]$.
\end{corollary}

\begin{lemma}
On $x_1 = 1$ and $x_2 \leq 1$, the function $\phi$ is minimised at 
\[
x_2 = \min \left( 1, \sqrt{\frac{b(Y+c)}{d(X+a)}} \right).
\]
On $x_2 = 1$ and $x_1 \leq 1$, the function $\phi$ is minimised at
\[
x_1 = \min \left( 1, \sqrt{\frac{a(Y+d)}{c(X+b)}} \right).
\]
\end{lemma}
\begin{proof}
This follows from writing
\begin{align*}
\phi(1, x_2) &= (X+a)(Y+c) + bd + \left( \frac{b(Y+c)}{x_2} + d(X+a) x_2 \right) \\
&= (X+a)(Y+c) + bd + 2\sqrt{bd(X+a)(Y+c)} \cosh \left(\ln\left(x_2  \sqrt{ \frac{d(X+a)}{b(Y+c)}}  \right)\right),
\end{align*}
which has a unique minimum at $x_2 = \sqrt{b(Y+C)/d(X+a)}$.
The proof for $\phi(x_1, 1)$ is similar.
\end{proof}

If we now replace $X = p_{tp} - p_{fp}$, $Y = \mathbb E(\tilde c)$, and $a = p_{fp}$, $b=p_{fn}$, $c = c_p$, $d = c_n$, and $x_i = \eta_i$ into $\phi$, then \cref{c:optimal-eta} holds.

\section{Coupling tau-leap and exact simulations}
\label{s:coupling-sims}

\Cref{a:tauleap,a:complete,a:map} demonstrate how to:
\begin{enumerate}
\item create a tau-leap low-fidelity simulation of a biochemical reaction network;
\item map a coarse-grained description of a unit-rate Poisson process into a fully described Poisson process;
\item map $M$ unit rate Poisson processes to an exact SSA trajectory involving $M$ reactions.
\end{enumerate}
These algorithms are used to simulate the low-fidelity models in \Cref{s:repressilator,ss:viral}, and to produce simulations from the high-fidelity model, conditional on the low-fidelity simulation, as described in \Cref{sss:coupling}.

\begin{algorithm}
\caption{Tau-leap}
\label{a:tauleap}
\begin{algorithmic}
\STATE{Input: interval leap $\tau$; propensity function $v(x) \in \mathbb R^M$; stochastic matrix $\nu \in \mathbb Z^{N \times M}$ with columns $\nu_j$; initial condition $x_0 \in \mathbb Z^N$; stopping time $T$.}
\STATE{Set $i=0$, $t_i=0$, and $x_i = x_0$.}
\WHILE{$t_i<T$}
	\STATE{Set $i = i+1$.}
	\FOR{$j = 1,\dots,M$}
		\STATE{Record $D_{ij} = \tau v_j(x_{i-1})$.}
		\STATE{Generate $P_{ij} \sim \mathrm{Poi}(D_{ij})$.}
	\ENDFOR
	\STATE{Set $t_i = t_{i-1}+\tau$.}
	\STATE{Set $x_i = x_{i-1} + \sum_{j=1}^{M} \nu_j P_{ij}$.}
\ENDWHILE
\RETURN Trajectory $(t_i,x_i)$ and $M$ partial Poisson processes described by $(D_{ij},P_{ij})$.
\end{algorithmic}
\end{algorithm}

\begin{algorithm}
\caption{Completing Poisson processes from tau-leap simulation.}
\label{a:complete}
\begin{algorithmic}
\STATE{Input: interval lengths $D_{ij}$; number of events $P_{ij}$ for $i=1,\dots,K$, and $j=1,\dots,M$.}
\FOR{$j=1,\dots,M$}
 \STATE{Set $P = 0$ and $D=0$}.
 \FOR{$i = 1,\dots,K$}
  \STATE{Generate $P_{ij}$ points, $d_{kj}$, (where $k=P+1,\dots,P+P_{ij}$) uniformly at random on the interval $(D, D+D_{ij}]$.}
  \STATE{Update $P = P + P_{ij}$.}
  \STATE{Update $D = D + D_{ij}$.}
 \ENDFOR
\ENDFOR
\RETURN $M$ independent unit-rate Poisson processes $(d_{kj})$ for $k=1,\dots,P_j$, where $P_j = \sum_i P_{ij}$.
\end{algorithmic}
\end{algorithm}

\begin{algorithm}
\caption{Map reaction Poisson processes to exact SSA trajectory.}
\label{a:map}
\begin{algorithmic}
\STATE{Input: Poisson processes $(d_{kj})$ for $j=1,\dots,M$ and $k=1,\dots,P_j$; propensity function $v(x) \in \mathbb R^M$; stochastic matrix $\nu \in \mathbb Z^{N \times M}$ with columns $\nu_j$; initial condition $x_0 \in \mathbb Z^N$; stopping time $T$.}
\STATE{Set $i=0$ as iteration, and $t_0 = 0$ as real-time.}
\STATE{Set $E_j = 1$ for $j=1,\dots,M$ as index of next event for each reaction channel.}
\STATE{Set $\sigma_j = 0$ for $j=1,\dots,M$ as reaction-time for each reaction.}
\WHILE{$t_i < T$}
	\STATE{Evaluate $v = v(x_i)$.}
	\WHILE{$E_j > P_j$ for any $j=1,\dots,M$}
		\STATE{Update $P_j = P_j + 1$.}
		\STATE{Generate $E \sim \mathrm{Exp}(1)$.}
		\STATE{Set $d_{P_j,j} = d_{P_j-1,j} + E$ to extend the unit level Poisson process for reaction $j$, if needed.}
	\ENDWHILE
	\STATE{Evaluate $\tau_j = (d_{E_j,j} - \sigma_j) /v_j$ for $j=1,\dots,M$ as real-time wait time until next event in each reaction channel.}
	\STATE{Evaluate $J = \argmin_j \tau_j$ as the next reaction firing to occur.}
	\STATE{Update $i = i+1$ as next iteration.}
	\STATE{Set $t_i = t_{i-1} + \tau_J$ to advance real-time.}
	\STATE{Set $x_i = x_{i-1} + \nu_J$ to update the state.}
	\STATE{Set $\sigma_j = \sigma_j + \tau_J v_j$ for $j=1,\dots,M$ to advance reaction-time for each reaction.}
	\STATE{Set $E_J = E_J + 1$ to specify the index of the next event.}
\ENDWHILE
\RETURN Exact SSA trajectory sequence $(t_i,x_i)$.
\end{algorithmic}
\end{algorithm}

\subsection{Repressilator}

In the example used in \Cref{s:repressilator}, the simulation $\tilde y \sim \tilde p(\cdot~|~\theta)$ of the low-fidelity model in \Cref{eq:repressilator} was generated following the algorithm described in~\cite{Cao2005}, developing the approach in~\cite{Gillespie2003}. 
Briefly, this algorithm uses tau-leaping~\cite{Gillespie2001} to discretise the time dimension, with adaptations to ensure that the molecule counts stay positive and to limit the rate of change of the propensities within a leap of length $\tau$.
The basic tau-leaping algorithm is given in \Cref{a:tauleap}: the key is that both a trajectory and a set of partially-described Poisson processes are output.
The algorithm we have used also adapts $\tau$: we have shared code at \url{https://github.com/tpprescott/mf-abc}.
After simulating $\tilde y$ from the low-fidelity model, if $U<\eta$ then \cref{a:sl} requires a simulation $y \sim p(\cdot~|~\theta)$ from the high-fidelity model.
Rather than simulating $y$ independently of $\tilde y$, we can instead define a coupling $p(\cdot~|~\tilde y,\theta)$, which will generate a coupled simulation $y \sim p(\cdot~|~\tilde y,\theta)$ from the high-fidelity model, conditional on $\tilde y$.

The coupling, $p(\cdot ~|~\tilde y,\theta)$, is defined by sharing the underlying unit-rate Poisson processes of each of the twelve reaction channels between the tau-leap and exact SSAs~\cite{Anderson2012,Lester2018,Lester2015}.
For each reaction channel, $j$, \Cref{a:tauleap} produces a coarse-grained description of the underlying random noise process.
This is a sequence of interval widths $D_{ij}$ and Poisson random numbers $P_{ij}$, for $i=1,\dots,N_j$ corresponding to the number of events in a unit rate Poisson process during that interval.
This coarse-grained description for each Poisson process was completed into a fully-described realisation using \Cref{a:complete} on each reaction's unit rate Poisson process.
For each $j=1,\dots,12$ and $i=1,\dots,N_j$, the $P_{ij}$ event times are placed uniformly on the interval of length $D_{ij}$, which are placed one after the other.
This completely describes a set of unit-rate Poisson processes (one for each reaction) that can then be mapped to an exact trajectory using \Cref{a:map}. 
An example of the output of this coupling is given in \cref{fig:couplingeg}.

\begin{figure}[tbhp]
\centering
\label{fig:couplingegfig}\includegraphics[width=\textwidth]{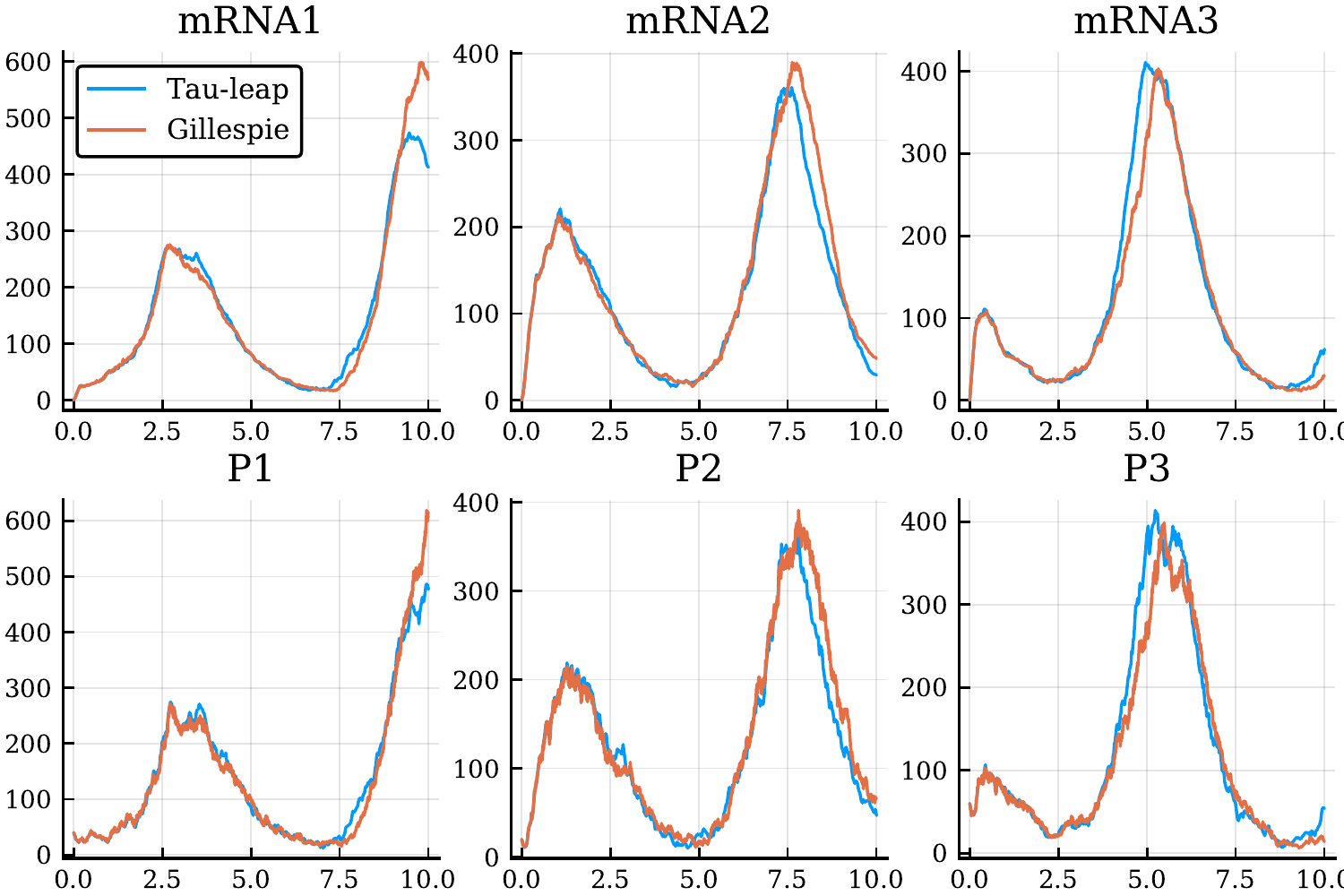}
\caption{Example of a tau-leap simulation $\tilde y$ (blue line), and a Gillespie simulation $y$ conditional on $\tilde y$ (orange line). The simulated molecule numbers of each mRNA (top row) and protein (bottom row) are shown. The parameter values used in these simulations are the nominal values $\alpha_0 =1$, $\beta = 5$, $\alpha = 1000$, $n = 2$, and $K_h = 20$, with the initial conditions $(m_1,m_2,m_3) = (0,0,0)$ and $(p_1,p_2,p_3) = (40,20,60)$. }
\label{fig:couplingeg}
\end{figure}

\subsection{Viral Kinetics}

The example in \Cref{ss:viral} adapts a low-fidelity model originally introduced in \cite{Haseltine2002}.
The reactions in \Cref{eq:viral-kinetics} comprise two fast reactions, \Cref{eq:fast1,eq:fast2}, which we approximate, and four slow reactions.
At the initial time $t=0$, we generate four random wait times corresponding to the slow reactions, as would happen in the Gillespie algorithm. 
We also generate one additional, deterministic, wait time
\[
\Delta t = \frac{-1}{k_5} \log( 1 - |1/\delta| ),
\]
where $\delta = (k_3 x_1(t_0)/k_5) - x_3(t_0)$ is the distance of $x_3$ from the value of its steady state mean, considering only the fast reactions in \Cref{eq:fast1,eq:fast2}. 
We set $\tau$ to be the minimum of these five wait times, and advance to $t_1 = \tau$.
If one of the random wait times is the minimum, then the corresponding reaction fires, as in the SSA. 
Alternatively, if the deterministic wait time is shortest, then the value of $x_3(t_1) = x_3(t_0) \pm 1$, depending on whether $\delta>0$ or $\delta<0$, respectively. 
In either case, all five wait times are then reduced by $\tau$, and a new wait time generated for the reaction whose wait time is reduced to zero.
The simulation of the low-fidelity model then continues iterating from $t_1$.
 
To simulate $y \sim p(\cdot~|~\theta)$ from the high-fidelity model, we can define a coupling $p(\cdot~|~\tilde y,\theta)$ that couples the simulation of the high-fidelity model to the simulation of the low-fidelity model.
In this example, four exact Poisson processes have already been produced corresponding to the slow reactions in \Cref{eq:viral-kinetics}. 
To simulate the high-fidelity model conditional on the simulation of the low-fidelity model, a further two independent unit-rate Poisson processes are produced, corresponding to the reactions in \Cref{eq:fast1,eq:fast2}, and mapped to an exact trajectory, as described in \Cref{a:map}.
The resulting coupled simulations are tightly correlated, as can be seen in the example in \cref{fig:viral-completion}:
the trajectories are essentially equal, with only the very fast stochastic fluctuations in \emph{struct} missing from the simulation of the low-fidelity model.
 
\begin{figure}[tbhp]
  \centering
  \includegraphics[width=\textwidth]{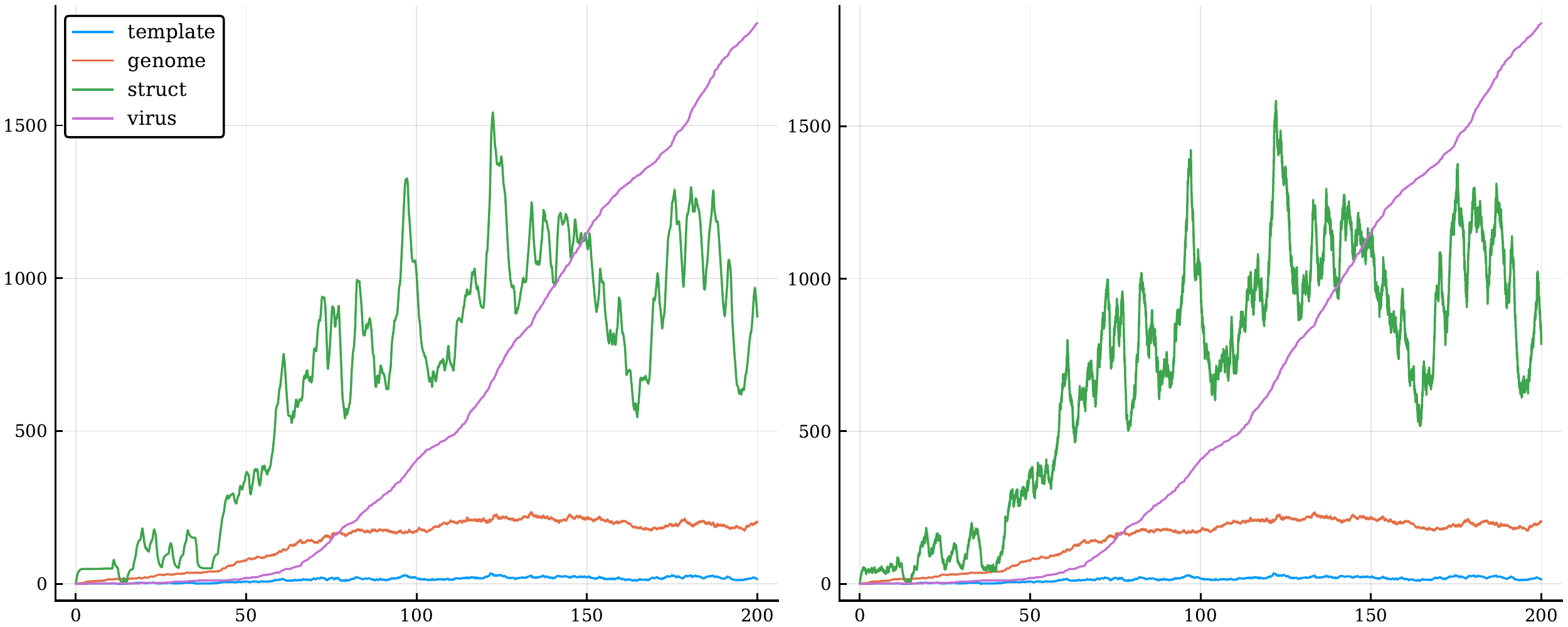}
  \caption{Example coupled simulations of (left) low-fidelity hybrid deterministic/stochastic model, and (right) high-fidelity SSA model, conditional on the simulation of the low-fidelity model. 
  Reaction rate parameters are the nominal values, $(1, 0.025, 100, 0.25, 1.9985, 7.5 \times 10^{-5})$, with initial conditions $(1,0,0,0)$.}
  \label{fig:viral-completion}
\end{figure}

\end{document}